\documentclass[11pt,oneside]{article}

\usepackage{color,soul}
\setulcolor{red}
\usepackage{scalerel}
\usepackage{latexsym}
\usepackage{amsmath}
\usepackage{amssymb}
\usepackage{amsthm}
\usepackage{wasysym}
\usepackage{cite}
\usepackage{graphicx}
\usepackage{color}
\usepackage{enumitem}
\usepackage{bm}
\usepackage[left=1in,right=1in,top=1in,bottom=1in]{geometry}
\usepackage{xspace}
\setlist[description]{font=\normalfont\itshape\textbullet\space}
\usepackage[all]{xy}
\usepackage[colorlinks=true]{hyperref}
\usepackage{tikz} 
\usetikzlibrary{positioning,chains,fit,shapes,calc}
\usepackage{float}
\usepackage{todonotes}

\renewcommand{\paragraph}[1]{\vspace{6pt} \noindent \textbf{#1}\xspace}

\numberwithin{equation}{section}
\newtheorem{theorem}{Theorem}[section]

\newtheorem{corollary}[theorem]{Corollary}
\newtheorem{lemma}[theorem]{Lemma}

\newtheorem{proposition}[theorem]{Proposition}

  {\end{tabular}\par\medskip}

\theoremstyle{definition}

\newtheorem{remark}[theorem]{Remark}

\newtheorem{definition}[theorem]{Definition}

\newtheorem{question}[theorem]{Open Question}

\newcommand{\SL}{\mathrm{SL}}
\newcommand{\GL}{\mathrm{GL}}
\newcommand{\K}{\mathbb{K}}
\newcommand{\F}{\mathbb{F}}

\newcommand{\C}{\mathbb{C}}
\newcommand{\R}{\mathbb{R}}
\newcommand{\N}{\mathbb{N}}

\newcommand{\poly}{\mathrm{poly}}

\newcommand{\M}{\mathrm{M}}
\newcommand{\T}{\mathrm{T}}
\renewcommand{\O}{\mathrm{O}}
\newcommand{\U}{\mathrm{U}}
\newcommand{\Sp}{\mathrm{Sp}}
\renewcommand{\S}{\mathrm{S}}

\newcommand{\E}{\mathrm{E}}

\newcommand{\tS}{\tens{S}}

\newcommand{\tuple}[1]{\mathbf{#1}}
\newcommand{\tens}[1]{\mathtt{#1}}
\newcommand{\spa}[1]{\mathcal{#1}}

\newcommand{\cE}{\spa{E}}
\newcommand{\cF}{\spa{F}}
\newcommand{\cG}{\spa{G}}
\newcommand{\cH}{\spa{H}}
\newcommand{\tA}{\tens{A}}
\newcommand{\tB}{\tens{B}}
\newcommand{\tE}{\tens{E}}

\newcommand{\vC}{\tuple{C}}
\newcommand{\vD}{\tuple{D}}

\newcommand{\vg}{\tuple{g}}

\newcommand{\vU}{\tuple{U}}

\newcommand{\aut}{\mathrm{Aut}}

\newcommand{\algprobm}[1]{{\sc #1}\xspace}
\newcommand{\GI}{\algprobm{GI}}
\newcommand{\DGI}{\algprobm{DGI}}
\newcommand{\GIlong}{\algprobm{Graph Isomorphism}}

\newcommand{\ThreeTIlong}{\algprobm{3-Tensor Isomorphism}}

\newcommand{\TIlong}{\algprobm{Tensor Isomorphism}}
\newcommand{\DeeTIlong}{\algprobm{$d$-Tensor Isomorphism}}

\newcommand{\CodeEqlong}{\algprobm{Code Equivalence}}

\newcommand{\Path}{\mathrm{Path}}

\newcommand{\cc}[1]{\mathsf{#1}}

\DeclareMathOperator{\rank}{rank}
\DeclareMathOperator{\diag}{diag}

\newcommand{\too}%
{\xrightarrow{\text{\raisebox{-3pt}{$\sim$}}\,}}

\hypersetup{
	pdftitle={Isomorphism problems over classical groups},
	pdfauthor={Zhili Chen, Joshua A. Grochow, Youming Qiao, Gang Tang, Chuanqi 
		Zhang},
	bookmarksnumbered=true,     
	bookmarksopen=true,         
	bookmarksopenlevel=1,       
	colorlinks=true,            
	pdfstartview=FitH,           
	pdfpagemode=UseOutlines,    
	pdfpagelayout=OneColumn
}

\title{
	On the complexity of isomorphism problems for tensors, groups, and 
	polynomials III: actions by classical groups\footnote{A preliminary version of this paper appeared at the 15th Innovations in Theoretical Computer Science Conference (ITCS 2024) \cite{ITCS24}.}
}

\author{
	Zhili Chen
	\footnote{Since August 2024, Centre for Quantum Technologies, National University of Singapore, Singapore ({\tt{chen.zhili@u.nus.edu}}). Research supported in part by the National Research Foundation, Singapore under its CQT Bridging Grant. }
	\footnote{Centre for Quantum Software and Information, University of Technology Sydney, Australia ({\tt{zhili.chen@student.uts.edu.au}}).}
	\and
	Joshua A. Grochow
	\footnote{Departments of Computer Science and Mathematics, University of Colorado, Boulder ({\tt{jgrochow@colorado.edu}}). Research supported by NSF CAREER grant CISE-2047756. } 
	\and
	Youming Qiao
	\footnote{Centre for Quantum Software and Information, University of Technology Sydney, Australia ({\tt{youming.qiao@uts.edu.au}}). Research supported in part supported by the Australian Research Council DP200100950 and LP220100332. }
	\and 
	Gang Tang
	\footnote{Centre for Quantum Software and Information, University of Technology Sydney, Australia ({\tt{gang.tang-1@student.uts.edu.au}}). Research supported in part by the Australian Research Council LP220100332 and the Sydney Quantum Academy, Sydney, NSW, Australia. }
	\and 
	Chuanqi Zhang
	\footnote{Centre for Quantum Software and Information, University of Technology Sydney, Australia ({\tt{chuanqi.zhang@student.uts.edu.au}}). Research supported in part by the Australian Research Council DP200100950 and the Sydney Quantum Academy, Sydney, NSW, Australia. }
}

\date{ \today
}

\begin{document}
\maketitle
\begin{abstract}
We study the complexity of isomorphism problems for $d$-way arrays, or tensors, under natural actions by classical groups such as orthogonal, unitary, and symplectic groups. These problems arise naturally in statistical data analysis and quantum information. We study two types of complexity-theoretic questions. First, for a fixed action type (isomorphism, conjugacy, etc.), we relate the complexity of the isomorphism problem over a classical group to that over the general linear group. Second, for a fixed group type (orthogonal, unitary, or symplectic), we compare the complexity of the isomorphism problems for different actions. 

Our main results are as follows. First, for orthogonal and symplectic groups acting on $3$-way arrays, the isomorphism problems reduce to the corresponding problems over the general linear group. Second, for orthogonal and unitary groups, the isomorphism problems of five natural actions on $3$-way arrays are polynomial-time equivalent, and the $d$-tensor isomorphism problem reduces to the $3$-tensor isomorphism problem for any fixed $d>3$. For unitary groups, the preceding result implies that LOCC classification of tripartite quantum states is at least as difficult as LOCC classification of $d$-partite quantum states for any $d$. Lastly, we also show that the graph isomorphism problem reduces to the tensor isomorphism problem over orthogonal and unitary groups. 
\end{abstract}

\section{Introduction}

Previously in \cite{FGS19,GQ21,TI2,GQT22,RST}, isomorphism problems of tensors, groups, and polynomials \emph{over direct products of general linear
	groups} were studied from the complexity-theoretic viewpoint. In particular, a complexity class $\cc{TI}$ (short for the Tensor Isomorphism class) was defined in \cite{GQ21}, and several isomorphism problems, including those for tensors, groups, and polynomials, were shown to be $\cc{TI}$-complete. The equivalence between polynomials and 3-tensors was shown subsequently but independently in \cite{RST}; some problems over 
products of general linear groups with monomial groups 
were also shown to be $\cc{TI}$-complete \cite{DAlconzo}.

In this paper, we study isomorphism problems of tensors, groups, and polynomials over
some classical groups, such as orthogonal, unitary, and symplectic groups, from the computational complexity viewpoint. There are several motivations to study tensor isomorphism over classical groups from statistical data analysis and quantum information. This introduction section is organised as follows. We will first review $d$-way arrays and some natural group actions on them in Section~\ref{subsec:five-action}, and describe motivations to study these actions over classical groups in Section~\ref{subsec:motivation}. We will then present our main results in Section~\ref{subsec:results}, and give an overview of the proofs in Section~\ref{sec:overview}. We conclude this introduction with a brief overview of the series of works this paper belongs to, a discussion on the results, and some open problems in Section~\ref{subsec:outlook}.

\subsection{Review of $d$-way arrays and some group actions on them}\label{subsec:five-action}
Let $\F$ be a field, and let $n_1, \dots, n_d\in \N$. For $n\in \N$, $[n]:=\{1, 2, \dots, n\}$. We use $\T(n_1\times\dots\times n_d, \F)$ to denote the linear space of $d$-way arrays with $[n_j]$ being the range of the $j$th index. That is, an element in $\T(n_1\times\dots\times n_d, \F)$ is of the form $\tA=(a_{i_1, \dots, i_d})$ where $\forall j\in[d]$, $i_j\in[n_j]$, and $a_{i_1, \dots, i_d}\in \F$. Note that $2$-way arrays are just matrices. Let $\M(n\times m, \F):=\T(n\times m, \F)$, and $\M(n, \F):=\M(n\times n, \F)$.

\begin{definition}\label{def:tensor}
	Let $\GL(n, \F)$ be the general linear group of degree $n$ over $\F$. We define an action of $\GL(n_1, \F)\times\dots\times \GL(n_d, \F)$ on $\T(n_1\times \dots\times n_d, \F)$, denoted as $\circ$, as follows. Let $\vg=(g_1, \dots, g_d)$, where $g_k\in\GL(n_k,\F)$ over $k\in[d]$. The action of $\vg$ sends $\tA=(a_{i_1, \dots, i_d})$ to $\vg\circ \tA=(b_{i_1, \dots, i_d})$, where $b_{i_1, \dots, i_d}=\sum_{j_1,\dotsc,j_d} a_{j_1, \dotsc, j_d} (g_1)_{i_1, j_1} (g_2)_{i_2, j_2} \dotsb (g_d)_{i_d,j_d}$.
\end{definition}

There are several group actions of direct products of general linear groups on $d$-way arrays, based on interpretations of $d$-way arrays as different multilinear algebraic objects. For example, there are three well-known natural actions on matrices: for $A \in \M(n,\F)$, (1) $(P,Q) \in \GL(n,\F) \times
\GL(n,\F)$ sends $A$ to $P^t A Q$, (2) $P \in \GL(n,\F)$ sends $A$ to $P^{-1} A
P$, and (3) $P \in \GL(n,\F)$ sends $A$ to $P^t A P$.
These three actions endow $A$ with different algebraic or geometric interpretations:
(1) a linear
map
from a vector space $V$ to another vector space $W$, (2) a linear map from $V$ to
itself, and (3) a bilinear map from $V\times V$ to $\F$. 

Analogously, there are five natural actions on $3$-way arrays, which we collect in the following definition (see \cite[Sec.~2.2]{GQ21} for more discussion of why these five capture all possibilities within a certain natural class).
\begin{definition}\label{def:five-action}
	We define five actions of (direct products of) general linear groups on $3$-way arrays. Note that in the following, $\circ$ is from Definition~\ref{def:tensor}.
	\begin{enumerate}
		\item Given $\tA\in\T(l\times m\times n,\F)$, $(P,Q,R)\in\GL(l,\F)\times\GL(m,\F)\times\GL(n,\F)$ sends $\tA$ to $(P,Q,R)\circ\tA$;
		\item Given $\tA\in\T(l\times l\times m,\F)$, $(P,Q)\in\GL(l,\F)\times\GL(m,\F)$ sends $\tA$ to $(P,P,Q)\circ\tA$;
		\item Given $\tA\in\T(l\times l\times m,\F)$, $(P,Q)\in\GL(l,\F)\times\GL(m,\F)$ sends $\tA$ to $(P,P^{-t},Q)\circ\tA$;
		\item Given $\tA\in\T(l\times l\times l,\F)$, $P\in\GL(l,\F)$ sends $\tA$ to $(P,P,P^{-t})\circ\tA$;
		\item Given $\tA\in\T(l\times l\times l,\F)$, $P\in\GL(l,\F)$ sends $\tA$ to $(P,P,P)\circ\tA$. 
	\end{enumerate}
\end{definition}
These five actions arise naturally by viewing $3$-way arrays as encoding, respectively: (1) tensors or matrix spaces (up to equivalence), (2) $p$-groups of class $2$ and exponent $p$, quadratic polynomial maps, or bilinear maps, (3) matrix spaces up to conjugacy, (4) algebras, and (5) trilinear forms or (noncommutative) cubic forms. For details on these interpretations, we refer the reader to \cite[Sec.~2.2]{GQ21}.

For a group $\cG$ acting on a set $S$, the isomorphism problem for this action asks to decide, given $s, t\in S$, whether $s$ and $t$ are in the same $\cG$-orbit. For example, \GIlong is the isomorphism problem for the action of the symmetric group $\S_n$ on $2^{\binom{[n]}{2}}$, the power set of the set of size-$2$ subsets of $[n]$. 

To help specify which of the five actions we are talking about, we use the following shorthand notation from multilinear algebra\footnote{See \cite{Lim21} for a nice survey of various viewpoints of tensors. For us, we have to start with the $d$-way array viewpoint, because we wish to study the relations between different actions, and the constructions are more intuitively described by examining the arrays.}. Let $U\cong \F^l$, $V\cong \F^m$ and $W\cong \F^n$. The dual space of a vector space $U$ is denoted as $U^*$. Then action (1) is referred to as $U\otimes V\otimes W$, (2) is $U\otimes U\otimes V$, (3) is $U\otimes U^*\otimes V$, (4) is $U\otimes U\otimes U^*$, and (5) is $U\otimes U\otimes U$. Note that from this shorthand notation, one can directly read off the action as in Definition~\ref{def:five-action} and vice versa.

\subsection{Motivations for isomorphism problems of $d$-way arrays over classical groups}\label{subsec:motivation}

The term ``classical groups'' appeared in Weyl's classic \cite{Wey97}, though there are multiple competing possibilities for what this term should mean formally \cite{Humphreys}. In this paper, we will be mostly concerned with \emph{groups consisting of elements that preserve a bilinear or sesquilinear form}, which include orthogonal groups $\O$, symplectic groups $\Sp$, and unitary groups $\U$, among others. As subgroups of $\GL$, they act naturally on $d$-way arrays. Note that for the orthogonal group $\O(n, \R)$, there are essentially three actions instead of five (because $P^{-t} = P$ for $P \in \O(n,\R)$). 

Actions of classical groups on $d$-way arrays have appeared in several areas of computational and applied mathematics \cite{Lim21}. In this subsection we examine some of these applications from statistical data analysis and quantum information.

\paragraph{Warm up: singular value decompositions.} Consider the action of $(A, B)\in \U(n, \C)\times \U(m, \C)$ on $C\in \M(n\times m, \C)$ by sending $C$ to $A^*CB$, where $A^*$ denotes the conjugate transpose of $A$. The orbits of this action are determined by the Singular Value Theorem, which states that every $C\in \M(n\times m, \C)$ can be written as $A^*DB$ where $A\in \U(n, \C)$, $B\in \U(m, \C)$, and $D\in \M(n\times m, \C)$ is a rectangular diagonal matrix. Furthermore, the diagonal entries of $D$ are non-negative real numbers, called the singular values of $C$. Similar results hold for $\O(n, \R)\times\O(m, \R)$ acting on $\R^n\otimes \R^m$.

This example illustrates that the orbit structure of $\U(n, \C)\times \U(m, \C)$ on $\M(n\times m, \C)$ is different from the action of $\GL(n, \C)\times\GL(m, \C)$ on $\M(n\times m, \C)$. Indeed, the former is determined by singular values (of which there are continuum many choices) and the latter is determined by rank (of which there are only finitely many choices).

\paragraph{Orthogonal isomorphism of tensors from data analysis.} The singular value decomposition is the basis for the Eckart--Young Theorem \cite{EY36}, which states that the best rank-$r$ approximation of a real matrix $C$ is the one obtained by summing up the rank-$1$ components corresponding to the largest $r$ singular values. To obtain a generalisation of such a result to $d$-way arrays, $d>2$, is a central problem in statistical analysis of multiway data \cite{DL08}.

Due to the close relation between singular value decompositions and orthogonal groups acting on matrices, it may not be surprising that the orthogonal equivalence of real $d$-way arrays is studied in this context \cite{DDV00,DL08,HU17,Sei18}. For example, one question is to study the relation between ``higher-order singular values'' and orbits under orthogonal group actions. From the perspective of the orthogonal equivalence of $d$-way arrays, such higher-order singular values are natural isomorphism invariants, though they do not characterise orbits as in the matrix case. In the literature, $d$-way arrays under orthogonal group actions are sometimes called Cartesian tensors \cite{CartesianTensor}.

\paragraph{Unitary isomorphism of tensors from quantum information.} We now turn to $\F=\C$ and consider the action of a product of unitary groups; such actions arise in at least two distinct ways in quantum information, which we highlight here: as LU or LOCC equivalence of quantum states, and as unitary equivalence of quantum channels.

In quantum information, unit vectors in $\T(n_1\times \dots\times n_d, \C)\cong \C^{n_1}\otimes\dots\otimes\C^{n_d}$ are called pure states, and two pure states are called locally-unitary (LU) equivalent, if they are in the same orbit under the natural action of $\vU:=\U(n_1, \C)\times\dots\times \U(n_d, \C)$ (where the $i$-th factor of the group acts on the $i$-th tensor factor). By Bennett \emph{et al.} \cite{LOCC}, the LU equivalence of pure states also captures their equivalence under local operations and classical communication (LOCC), which means that LU-equivalent states are inter-convertible by reasonable physical operations.

A completely positive map is a function $f:\M(n, \C)\to\M(n, \C)$ of the form $f(A)=\sum_{i\in[m]}B_iAB_i^*$ for some complex matrices $B_i \in \M(n, \C)$; quantum channels are given precisely by the completely positive maps that are also ``trace-preserving'', in the sense that $\sum_{i\in[m]}B_i^*B_i=I_n$. Two tuples of matrices $(B_1, \dots, B_m)$ and $(B_1', \dots, B_m')$ define the same completely positive map if and only if there exists $S=(s_{i,j})\in\U(m, \C)$ such that $\forall i\in[m]$, $B_i=\sum_{j\in[m]}s_{i,j}B_j'$ \cite[Theorem 8.2]{NC00}. And two quantum channels $f, g:\M(n, \C)\to\M(n, \C)$ are called unitarily equivalent if there exists $T\in \U(n, \C)$ such that for any $A\in \M(n, \C)$, $T^*f(A)T=g(T^*AT)$. Thus, two matrix tuples $(B_1, \dots, B_m)$ and $(B_1', \dots, B_m')$ define the unitarily equivalent quantum channels if and only if their corresponding $3$-way arrays in $\T(n\times n\times m, \C)$ are in the same orbit under a natural action of $\U(n, \C)\times \U(m, \C)$.

\paragraph{Classical groups arising from \CodeEqlong.} Classical groups may appear even when we start with general linear or symmetric groups. Here is an example from code equivalence. Recall that the (permutation linear) code equivalence problem asks the following: given two matrices $A, B\in \M(d\times n, q)$, decide if there exist $C\in\GL(d, q)$ and $P\in \S_n$, such that $A=CBP$. One algorithm for this problem, under some conditions on $A$ and $B$, from \cite{BOS19} goes as follows. Suppose it is the case that $A=CBP$. Then $AA^t=CBPP^tB^tC^t=CBB^tC^t$. This means that $AA^t$ and $BB^t$ are congruent. Assuming that $AA^t$ and $BB^t$ are full-rank, then up to a change of basis, we can set that $AA^t=BB^t=:F$, so any such $C$ must lie in a classical group preserving the form $F$. We are then reduced to the problem of asking whether $A$ and $B$ are equivalent up to some $C$ from a classical group and some $P$ from a permutation group. This problem, as shown in \cite{BOS19}, reduces to \GIlong. 

\paragraph{Some preliminary remarks on the algorithms for \TIlong over classical groups.} 
Although we show that \algprobm{Orthogonal TI} and \algprobm{Unitary TI}  are still \GI-hard (Theorem \ref{thm: gi to ti}), from the current literature it seems that orthogonal and unitary isomorphism of tensors are easier than general-linear isomorphism. There are currently two reasons for this: the first is mathematical, and the second is based on practical algorithmic experience, which we now discuss.

One mathematical reason why these problems may be easier is that there are easily computable isomorphism invariants for such actions, while such invariants are not known for general-linear group actions. Here is one construction of an effective invariant in the unitary case. From $\tA=(a_{i,j,k})\in\T(n\times n\times n, \C)$, construct its matrix flattening $B=(b_{i,j})\in \M(n\times n^2, \C)$, where $b_{i,j\cdot n+k}=a_{i,j,k}$. Then it can be verified easily that $|\det(BB^*)|$ is a polynomial-time computable isomorphism invariant for the unitary group action $\U(n, \C)\times\U(n, \C)\times\U(n, \C)$. However, it is not known whether such isomorphism invariants for the general linear group action exist---if they did, they would break the pseudo-random assumption for this action proposed in \cite{JQSY19}.

Practically speaking, current techniques seem much more effective at solving tensor isomorphism-style problems over the orthogonal group than over the general linear group. It is not hard to formulate \TIlong and related problems over general linear and some classical groups as solving systems of polynomial equations. Motivated by cryptographic applications \cite{TangDJPQS22}, we chose a $\cc{TI}$-complete problem \algprobm{Alternating Trilinear Form Isomorphism} \cite{GQT22}, and carried out experiments using the Gr\"obner basis method for this problem, implemented in Magma \cite{Magma}. For some details of these experiments see Appendix~\ref{app:groebner}. We fixed the underlying field order as $32771$ (a large prime that is close to a power of 2). Over the general linear group for $n=7$, the solver ran for about 3 weeks on a server, eating 219.7GB memory, yet still did not complete with a solution. Over the orthogonal group for odd $n$, the data are shown in Table~\ref{tb:GB}. In particular, the solver returns a solution for $n=21$ in about 3.6 hours, a sharp contrast to the difficulty met when solving the problem under the general linear group action. 
\begin{table}[h]
	\centering
	\begin{tabular}{|c|c|c|c|c|c|c|c|c|}
		\hline
		$n$ & 7 & 9 & 11 & 13 & 15 & 17 & 19 & 21 \\
		\hline
		Time (in s) & 0.396 & 5.039 & 37.120 & 140.479 & 524.520 & 1764.179 & 4720.129 & 12959.799 \\
		\hline
	\end{tabular}
	\caption{The experiment results of the Gr\"obner basis method to solve the problem of isomorphism of alternating trilinear forms under the action of the orthogonal group. }
	\label{tb:GB}
\end{table}

\subsection{Our results}\label{subsec:results}

In this paper we study the complexity-theoretic aspects of \TIlong under classical groups. We focus on the following two types of questions:
\begin{enumerate}
	\item Consider two classical groups $\cG$ and $\cH$, and fix the way they act on $d$-way arrays. What are the relations between the isomorphism problems defined by these groups?
	\item Fix a classical group $\cG$, and consider its different actions on $d$-way arrays. What are the relations between the isomorphism problems defined by these actions?
\end{enumerate}

Questions of the first type were implicitly studied in \cite{HQ20,GQ21,TI-IV} for some classes of $d$-way arrays, with the groups being either general linear or symmetric groups. For example, starting from a graph $G$, one can construct a $3$-way array $\tA_G$ encoding this graph following Edmonds, Tutte and Lov\'asz \cite{Tut47,Edm65,Lov79}, and it is shown in \cite{HQ20} that $G$ and $H$ are isomorphic (a notion based on the symmetric groups $S_n$) if and only if $\tA_G$ and $\tA_H$ are isomorphic (under a product of general linear groups). 

Questions of the second type were studied in \cite{FGS19,GQ21} for $\GL$. For example, one main result in \cite{FGS19,GQ21} is to show the polynomial-time equivalence of the five isomorphism problems for $3$-way arrays under (direct products of) general linear groups (cf. Section~\ref{subsec:five-action}).

Still, to the best of our knowledge, these types of questions have not been studied for orthogonal, unitary, and symplectic groups, which are the focus on this paper. 

\paragraph{Results on relations between different groups.} Our first group of results shows that isomorphism problems of tensors under classical groups are sandwiched between the celebrated \GIlong problem and the more familiar \TIlong problem under $\GL$. We use $\S_n$ to denote the symmetric group of degree $n$, and view $\S_n$ as a subgroup of $\GL(n, \F)$ naturally via permutation matrices. We use $\leq$ to denote the subgroup relation. When we say ``reduces'', briefly, we mean: polynomial-time computable kernel reductions \cite{FortnowGrochowPEq} (there is a polynomial-time function $r$ sending $(A,B)$ to $(r(A), r(B))$, such that the map $(A,B) \mapsto (r(A), r(B))$ is a many-one reduction of isomorphism problems), that are typically polynomial-size projections (``$p$-projections'') in the sense of Valiant \cite{Val79}, functorial (on isomorphisms), and containments in the sense of the literature on wildness. Some reductions that use a non-degeneracy condition may not be $p$-projections. See \cite[Sec.~2.3]{GQ21} for details on these notions. 

\begin{theorem}\label{thm:graph-iso}
	Suppose a group family $\cG=\{\cG_n\}$ satisfies that $\S_n\leq \cG_n\leq \GL(n, \F)$, where here $S_n$ denotes the group of $n \times n$ permutation matrices. Then \GIlong reduces to \algprobm{Bilinear Form $\cG$-Pseudo-isometry}, that is, the isomorphism problem for the action of $\cG(U)\times\cG(V)$ on  $U\otimes U\otimes V$. 
\end{theorem}

Let $\cG_n\leq \GL(n, \F)$. We say that $\cG_n$ \emph{preserves a bilinear form}, if there exists some $A\in \M(n, \F)$, such that $\cG_n=\{T\in \GL(n, \F) \mid T^t AT=A\}$. For example, orthogonal and symplectic groups are defined as preserving full-rank symmetric and skew-symmetric forms. 
\begin{theorem}\label{thm:toGL}
	Let $\cG=\{\cG_n\mid \cG_n\leq\GL(n, \F)\}$ be a group family preserving a polynomial-time-constructible family of bilinear forms,\footnote{That is, the function $\Phi \colon \mathbb{N} \to \M(n,\F)$ giving a matrix for the form preserved by $\cG_n$ is computable in polynomial time. We note that no such restriction was needed in Theorem~\ref{thm:graph-iso}.} and consider one of the five actions of $\GL$ on $3$-way arrays in Definition~\ref{def:five-action}. The restricted $\cG$-isomorphism problem for this action reduces to the $\GL$-isomorphism problem for this action. 
\end{theorem}

\begin{remark}
	Recall from Section~\ref{subsec:motivation} that the orthogonal equivalence of matrices (determined by singular values) is more involved than the general-linear equivalence of matrices (determined by ranks) over $\R$. By a counting argument, 
	there is unconditionally no polynomial-size kernel reduction \cite{FortnowGrochowPEq} (mapping matrices to matrices) from \algprobm{Orthogonal Equivalence of Matrices} to \algprobm{General Linear Equivalence of Matrices}.
	In contrast, Theorem~\ref{thm:toGL} shows that for $3$-way arrays, orthogonal isomorphism does reduce to general-linear isomorphism. 
\end{remark}

\paragraph{Results on relations between different actions.} Our second group of results is concerned with different actions of the same group on $d$-way arrays. Our main results are for the real orthogonal groups and complex unitary groups; we discuss some difficulties encountered with symplectic groups in Section~\ref{subsec:outlook}, and leave open the questions for more general bilinear-form-preserving groups.

We begin with the five actions in Definition~\ref{def:five-action}. 
\begin{theorem}\label{thm:5actions}
	Let $\cG$ be either the unitary over $\C$ or orthogonal over $\R$ group family. Then the five isomorphism problems corresponding to the five actions of $\cG$ on $3$-way arrays in Definition~\ref{def:five-action} are polynomial-time equivalent to one another. 
\end{theorem}

Our second result in this group is a reduction from $d$-way arrays to $3$-way arrays. 
\begin{theorem}\label{thm:dto3}
	Let $\cG$ be the unitary  over $\C$ or orthogonal over $\R$ group family. For any fixed $d \geq 1$, \algprobm{$d$-Tensor $\cG$-Isomorphism} reduces to \algprobm{3-Tensor $\cG$-Isomorphism}.
\end{theorem}

\paragraph{An application in quantum information.} As introduced in Section~\ref{subsec:motivation}, LU equivalence, characterises the equivalence of quantum states under local operations and classical communication (LOCC). We refer the interested reader to the nice paper \cite{LOCC_survey} for the LOCC notion, as well as the classification of three-qubit states based on LOCC \cite{ABLS01}. 

By the work of Bennett \emph{et al.} \cite{LOCC}, LOCC equivalence of pure quantum states is the same as the equivalence of unit vectors in $V_1\otimes V_2\otimes \dots \otimes V_d$ where $V_i$ are vector spaces over $\C$. Our Theorem~\ref{thm:dto3} can then be interpreted as saying that classifying tripartite quantum states under LOCC equivalence is as difficult as classifying $d$-partite quantum states. This may be compared with the result in \cite{ZLQ18}, which states that classifying $d$-partite states reduces to classifying tensor networks of tripartite or bipartite tensors. (We note that the analogous result for SLOCC, via the general linear group action, was shown in \cite{GQ21}; in the next section we discuss how our proof here differs from the one there.)

\subsection{Overview of the proofs of main results}\label{sec:overview}

In the following, we present proof outlines for Theorems~\ref{thm:graph-iso}, \ref{thm:toGL}, \ref{thm:5actions}, and \ref{thm:dto3}. While their proofs are inspired the strategies of previous results \cite{FGS19,GQ21,LQWWZ22}, new technical ingredients are indeed needed, such as the Singular Value Theorem, and a certain Krull--Schmidt type result for matrix tuples under unitary group actions. We also wish to highlight that, Theorem~\ref{thm:dto3}  requires not only using a quiver different from that in the proof of \cite[Theorem 1.2]{GQ21}, but also a completely new 
and much simpler argument. 

\paragraph{About Theorem~\ref{thm:graph-iso}.} For Theorem~\ref{thm:graph-iso}, we start with \algprobm{Directed Graph Isomorphism} (\DGI), which is $\cc{GI}$-complete. We then use a natural construction of $3$-way arrays from directed graphs as recently studied in \cite{LQWWZ22}, which takes an arc $(i, j)$ and constructs an elementary matrix $\E_{i, j}$. By \cite[Observation 6.1, Proposition 6.2]{LQWWZ22}, \DGI reduces to the isomorphism problem of $U\otimes U\otimes W$ under $\GL(U)\times \GL(W)$. Theorem~\ref{thm:graph-iso} is shown by observing that the proofs of \cite[Observation 6.1, Proposition 6.2]{LQWWZ22} carry over to all subgroups of $\GL(U)$ and $\GL(W)$ that contain the corresponding symmetric groups. 

\paragraph{About Theorem~\ref{thm:toGL}.} For Theorem~\ref{thm:toGL}, let us consider the isomorphism problem of $U\otimes V\otimes W$ under $\O(U)\times\O(V)\times\O(W)$. Let $a=\dim(U)$, $b=\dim(V)$, and $c=\dim(W)$. That is, given $\tA, \tB\in \T(a\times b\times c, \F)$, we want to decide if there exists $(R, S, T)\in\O(a, \F)\times\O(b, \F)\times\O(c, \F)$, such that $(R, S, T)\circ \tA=\tB$. Our goal is to reduce this problem to an isomorphism problem of $U'\otimes V'\otimes W'$ under $\GL(U')\times\GL(V')\times\GL(W')$. The idea is to encode the requirements of $R, S, T$ being orthogonal by adding identity matrices. We then construct tensor systems $(\tA, I_1, I_2, I_3)$ and $(\tB, I_1, I_2, I_3)$ where $I_1\in \M(a, \F)$, $I_2\in \M(b, \F)$, and $I_3\in \M(c, \F)$ are the identity matrices, and the goal is to decide if there exists $(R, S, T)\in \GL(a, \F)\times\GL(b, \F)\times\GL(c, \F)$ such that $(R, S, T)\circ \tA=\tB$, $R^t R=I_1$, $S^t S=I_2$, and $T^t T=I_3$. Such a problem falls into the tensor system framework in \cite{FGS19}; a main result of \cite[Theorem 1.1]{FGS19} can be rephrased as a reduction from \algprobm{Tensor System Isomorphism} to \ThreeTIlong.

\paragraph{About Theorem~\ref{thm:5actions}.} For Theorem~\ref{thm:5actions}, polynomial-time reductions for the five actions under $\GL$ were devised in \cite{FGS19,GQ21}. The main proof technique is a gadget construction, first proposed in \cite{FGS19}, which we call the Furtony--Grochow--Sergeichuk gadget, or FGS gadget for short. Roughly speaking, this gadget has the effect of reducing isomorphism over block-upper-triangular invertible matrices to that over general invertible matrices. We will explain why this is useful for our purpose, and the structure of this gadget, in the following. 

First, let us examine a setting when we wish to restrict to consider only block-upper-triangular matrices. Suppose we wish to reduce isomorphism of $U\otimes V\otimes W$ to that of $U'\otimes U'\otimes W'$. One naive idea is to set $U'=U\oplus V$ and $W'=W$, and perform the following construction. Let $\tA\in \T(\ell\times m\times n, \F)$, and take the frontal slices of $\tA$ as $(A_1, \dots, A_n)\in\M(\ell\times m, \F)$. Then construct $\tA' = (A_1', \dots, A_n')\in\M(\ell+m, \F)$, where $A_i'=\begin{bmatrix}
0 & A_i \\
-A_i^t & 0
\end{bmatrix}$, and let the corresponding 3-way array be $\tA'\in\T((\ell+m)\times(\ell+m)\times n, \F)$. Similarly, starting from $\tB\in\T(\ell\times m\times n, \F)$, we can construct $\tB'$ in the same way. The wish here is that $\tA$ and $\tB$ are unitarily isomorphic in $U\otimes V\otimes W$ if and only if $\tA'$ and $\tB'$ are unitarily isomorphic in $U'\otimes U'\otimes W'$. It can be verified that the only if direction holds easily, but the if direction is tricky. This is because, if we start with some isomorphism $(R, S)\in \U(U')\times \U(W')$ from $\tA'$ to $\tB'$, $R$ may mix the $U$ and $V$ parts of $U'$.

This problem---more generally, the problem of two parts of the vector space potentially mixing in undesired ways---is solved by the FGS gadget, which attaches identity matrices of appropriate ranks to prevent such mixing. Figure~\ref{fig:3-tensor_isometry} is an illustration from \cite{GQ21}.
\begin{figure}[!htbp]
	\[
	\xymatrix@R=4pt@C=3pt{
		&&&&& *{} \ar@{-}'[rrrrr] 
		\ar@{-}'[dddd]|{\ell}'[dddddr]|{n}'[dddddrrrrrr]'[ddddrrrrr]'[rrrrr]'[rrrrrdr]'[rrrrrdrdddd]
		\ar@{}'[ddddrrrrr]|{A} &&&&& *{} & 
		\\ 
		&&&&&&&&&&& *{} \ar@{-}'[rrrr]'[rrrrdr]'[rrrrdrd]'[rrrrd]'[d]'[drd]'[drdrrrr] &&&& 
		*{} \ar@{-}'[d] &&& 
		\\ 
		&&&&&&&&&&& *{} \ar@{.}'[rdrrrr]|{I_{m+1}} &&&& *{} & 
		*{} 
		\\ 
		&&&&&&&&&&& *{} \ar@{--}'[rd]  
		& *{} \ar@{-}'[d]'[drd]'[drdrrrr]'[drrrrr]'[rrrr]'[rrrrd]'[rrrrddr] &&&& *{}  %
		\\ 
		*{} \ar@{-}'[rrrrr]|{\ell} \ar@{-}'[dddd]|{m} \ar@{}'[ddddrrrrr]|{-A^t}
		&&&&& *{} \ar@{-}'[rrrrr]|{m} \ar@{-}'[dddd] &&&&& *{}  
		&*{} \ar@{--}'[rdrd]
		& *{} \ar@{-}'[rrrr] \ar@{.}'[rrrrrd]|{I_{m+1}} &&&& *{}  
		& *{} 
		\\ 
		&&&&&& *{} \ar@{-}'[dddd] 
		\ar@{--}'[rdrdrd]'[rdrdrdrrrrr]'[rdrdrdrdrrrrrrd][rdrdrdrdrrrrrrdddddddddd] %
		& \ar@{--}'[rdrdrd]'[rdrdrdddddddd] 
		&&&& *{} \ar@{--}'[rdrd] 
		&& *{} \ar@{}'[drrrr]|{\ddots} &&&& *{} 
		\\ 
		&&&&&& *{} 
		\ar@{--}'[rdrdrd]'[rrrdddrrrrrrrrr] 
		&&&&&&& *{} &&&& *{} 
		\\ 
		&&&&&&&&&&&&& *{} \ar@{-}'[rrrr] 
		\ar@{-}'[u]'[urrrr]'[rrrr]'[rrrrrd]'[rrrrr]'[urrrr] 
		\ar@{-}'[rd]'[rdrrrr] 
		\ar@{.}'[rdrrrr]|{I_{m+1}} 
		& & & & *{} & 
		*{} 
		\\ 
		*{} \ar@{-}'[rrrrr] \ar@{-}'[rd]'[rdrrrrr] 
		& & & 
		& &  *{} \ar@{-}'[rd] 
		& & & & *{} \ar@{--}'[dddd] 
		& & && &  *{} & & & &  *{} 
		\ar@{-}'[rrrrr]'[rrrrrd]'[d]'[ddr]'[ddrrrrrr]'[rrrrrd]'[rrrrr]'[rrrrrdr]'[rrrrrdrd]
		\ar@{-}'[d] 
		& & & & & *{} 
		\\ 
		& *{} \ar@{-}'[ddd]'[dddrd]'[dddrdr]'[dddr]'[r]'[rrd]'[rrdddd] 
		& *{} \ar@{.}'[rdddd] 
		& & & *{}  \ar@{--}'[rdrd] 
		& *{} 
		\ar@{--}'[rdrd] 
		&&& *{} 
		&&&&&&&&& *{} \ar@{.}'[rdrrrrr]|{I_{3m+2}} &&&&& *{}  
		& *{} 
		\\ 
		&&& *{} \ar@{}'[rrrrddd]|{\ddots} 
		&&&&&&&&&&&&&&&& *{}  \ar@{-}'[d]'[drrrrr]'[rrrrr]'[rrrrrdr]'[rrrrrdrd]'[drd]'[d] 
		&&&&&*{} 
		\\ 
		&&&&&&&*{} \ar@{-}'[ddd]'[dddrd]'[dddrdr]'[dddr]'[r]'[rrd]'[rrdddd] \ar@{-}'[r] %
		& *{} \ar@{.}'[rdddd] 
		&&&&&&&&&&& *{} \ar@{.}'[rrrrrdr]|{I_{3m+2}} 
		&&&&&*{} \ar@{-}'[rd] 
		&*{} 
		\\ 
		&*{} \ar@{-}'[r] &*{} 
		&&&&&&&*{} 
		\ar@{--}'[rdrdrd]'[rdrdrdrrrrrrrrrr] 
		&&&&&&&&&&&*{} \ar@{}'[drrrrrr]|{\ddots} 
		&&&&&*{} 
		\\ 
		&&*{}&*{} 
		&&&&&&&&&&&&&&&&&&*{} 
		\ar@{-}'[rrrrr]'[rrrrrd]'[d]'[ddr]'[ddrrrrrr]'[rrrrrd]'[rrrrr]'[rrrrrdr]'[rrrrrdrd]
		\ar@{-}'[d] 
		&&&&&*{} 
		\\ 
		&&&&&&&*{} \ar@{-}'[r] & *{} 
		&&&&&&&&&&&&&*{} \ar@{.}'[rdrrrrr]|{I_{3m+2}}&&&&&*{}
		& *{} 
		\\ 
		&&&&&&&&*{} 
		&*{}  \ar@{-}'[dddd]'[ddddrd]'[ddddrdr]'[ddddr]'[r]'[rrd]'[rrddddd] \ar@{-}'[r]  %
		&*{} \ar@{.}'[rddddd] 
		&&&&&&&&&&&&*{} &&&&&*{} 
		\\ 
		&&&&&&&&&&&*{} 
		&&&&&&&&&&&&&&&&&&&&&&&&&&&&&&&&&&&
		\\ 
		&&&&&&&&&&& \ar@{}'[rdrdrd]|{\ddots} 
		&&&&&&&&&&&&&&&&&&&&&&&&&&&&&&&&&&&&&& 
		\\ 
		&&&&&&&&&&&&&&*{} \ar@{-}'[dddd]'[ddddrd]'[ddddrdr]'[ddddr]'[r]'[rrd]'[rrddddd] 
		\ar@{-}'[r]  & *{} \ar@{.}'[rddddd] 
		&&&&&&&&&&&&&&&&&&&&&&&&&&&&&&
		\\ 
		&&&&&&&&&*{} \ar@{-}'[r] & *{} 
		&&&&&&*{} 
		&&&&&&&&&&&&&&&&&&&&&&&&&&&&& 
		\\ 
		&&&&&&&&&&*{} & *{} 
		&&&&&&&&&&&&&&&&&&&&&&&&&&&&&&&&&&&&&&
		\\ 
		&&&&&&&&&&&&&*{} &*{} 
		&&&&&&&&&&&&&&&&&&&&&&&&&&&&&&&
		\\ 
		&&&&&&&&&&&& &&*{} \ar@{-}'[r] &*{}  
		&&&&&&&&&&&&&&&&&&&&&&&&&&&&&&&
		\\ 
		&&&&&&&&&&&&&&&*{} &*{} 
		&&&&&&&&&&&&&&&&&&&&&&&&&&&&&&&
	}
	\]
	\caption{ \label{fig:3-tensor_isometry} Pictorial representation of the reduction 
		for Theorem~\ref{thm:5actions}; credit for the figure goes to the authors of \cite{GQ21}, reproduced here with their permission.}
\end{figure}
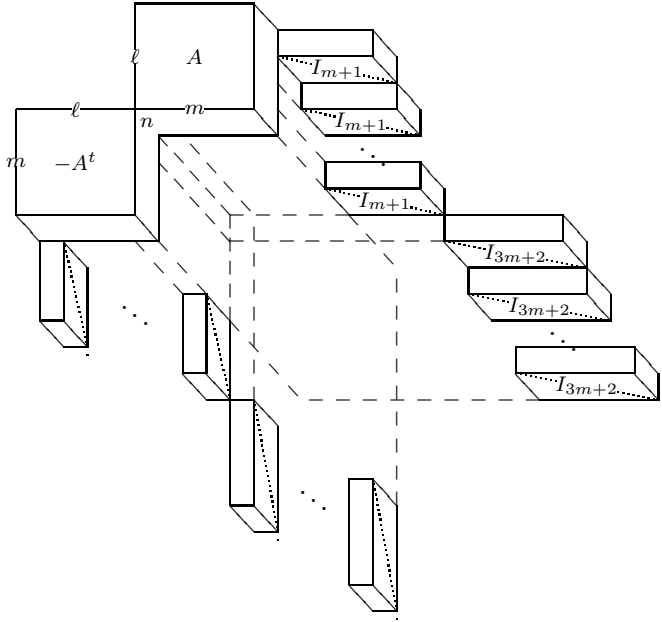
It can be verified that, because of the identity matrices $I_{m+1}$ and $I_{3m+2}$, an isomorphism $R$ in the $U'$ part has to be block-upper-triangular, and the blocks would yield the desired isomorphism for the $U$ and $W$ parts. 

This was done for the general linear group case in \cite{GQ21}. For the unitary group case, this almost goes through, because if a unitary matrix is block-upper-triangular, then it is actually block-diagonal, and the blocks are unitary too. Still, some technical difficulties remain. For example, now the gadgets cause some problem for the only if direction (which was easy in the $\GL$ case), so we must verify carefully that the added gadgets allow for extending the original orthogonal or unitary transformations to bigger ones. As another example, the proof in \cite{FGS19} relies on the Krull--Schmidt theorem for quiver representations (under general linear group actions). Fortunately, in our context we can replace that with a result of Sergeichuk \cite[Theorem 3.1]{Ser98} so that the proof can go through. Finally, we also require the use of the Singular Value Theorem to handle certain degenerate cases.

\paragraph{About Theorem~\ref{thm:dto3}.} For Theorem~\ref{thm:dto3}, at a high level we follow the strategy of reduction from \DeeTIlong to \ThreeTIlong from \cite{GQ21}, but we find that the construction there does not quite work in the setting of orthogonal or unitary group actions. As in \cite{GQ21}, we shall reduce \DeeTIlong to \algprobm{Algebra Isomorphism}, which reduces to \ThreeTIlong by Theorem~\ref{thm:5actions}. As in \cite{GQ21}, we also use path algebras. However, they use Mal'cev's result on the conjugacy of the Wedderburn complements of the Jacobson radical, and this result seems not to hold if we require the conjugating matrix to be orthogonal or unitary. To get around this, our main technical contribution is to develop a related but in fact \emph{simpler} path algebra construction, that avoids the use of the aforementioned deep algebraic results, and works not only in the $\GL$ setting, but extends to the orthogonal and unitary settings as well. This then gives us the reduction from \algprobm{$d$-Tensor Orthogonal Isomorphism} to \algprobm{Orthogonal Algebra Isomorphism}, and similarly in the unitary case.

\subsection{Summary and future directions}\label{subsec:outlook}

\paragraph{Context within recent developments on the complexity of \algprobm{Tensor Isomorphism}.} 
Following \cite{GQ21,TI2}, this paper contributes to building up the complexity theory around \algprobm{Tensor Isomorphism} and closely related problems. That is, \cite{GQ21} introduced $\cc{TI}$-completeness and showed that many isomorphism problems, under the action of a product of general linear groups, were $\cc{TI}$-complete. Then \cite{TI2} focused on applications of tensor techniques for reductions around \algprobm{$p$-Group Isomorphism}. Several recent works further enrich this theory, such as \cite{GQT22,DAlconzo} showing more problems to be $\cc{TI}$-complete. 

In \cite{TI-IV}, more efficient reductions between the five actions by general linear groups are designed. More specifically, several key reductions in \cite{GQ21} incur quadratic increases in the side lengths of the resulting tensors (or bilinear maps). In \cite{TI-IV}, new reductions are devised that incur only linear increases, which leads to several applications and further results. However, it seems difficult to adapt these linear-length reductions to the classical group actions as studied in this paper, due to the ``cancellation'' step in \cite{TI-IV} in which matrices corresponding to elementary row and column operations (therefore not unitary or orthogonal) are used. It is an interesting open problem to devise linear-length reductions between the five actions for some classical groups.


\paragraph{Some remarks on our results and techniques for more matrix groups.} In this paper, we examine isomorphism problems of $d$-way arrays under various actions of different subgroups of the general linear group from a complexity-theoretic viewpoint. We show that for $3$-way arrays, the isomorphism problems over orthogonal and symplectic groups reduce to that over the general linear group. We also show that for orthogonal and unitary groups, the five isomorphism problems corresponding to the five natural actions  are polynomial-time equivalent, and \DeeTIlong reduces to \ThreeTIlong.

As seen in Section~\ref{sec:overview}, the proof strategies of our results are adapted from previous works \cite{FGS19,GQ21,LQWWZ22}, although certain non-trivial adaptations were necessary, especially for the proofs of Theorem~\ref{thm:5actions} and~\ref{thm:dto3}, beyond careful examinations of previous proofs. Interestingly, in extending the proof strategies from these previous works to our main results, we also encountered some obstacles that would seem are more generally obstacles to reaching a uniform result for all classical groups. For example, the reduction from orthogonal and symplectic to general linear seems not work for unitary---the standard linear-algebraic gadgets have no way to force complex conjugation---and the reductions between the five actions on $3$-way arrays seem not work for symplectic. One stumbling block (pun intended) in the symplectic case is that even a symplectic block-\emph{diagonal} matrix (let alone a symplectic block-triangular matrix) need not have its individual blocks be symplectic. For example, the matrix $A \oplus B$, with $A,B$ both $n \times n$, is symplectic iff $AB^t = I$. 

\paragraph{Complexity classes $\cc{TI}_\cG$.} To put some of these remaining questions in a larger framework, we introduce a notation that highlights the role of the group doing the acting. Previously in computational complexity, the most studied isomorphism problems are over symmetric groups (such as \GIlong) and over general linear groups (such as tensor, group, and polynomial isomorphism problems). The former leads to the complexity class $\cc{GI}$ \cite{KST93}, and the latter leads to the complexity class $\cc{TI}$ \cite{GQ21}. Based on Theorems~\ref{thm:5actions} and~\ref{thm:dto3}, it may be interesting to define $\cc{TI}_\cG$, where $\cG$ is a family of matrix groups, consisting of all problems polynomial-time reducible to the 3-tensor isomorphism problem over $\cG$. Let $\S$, $\GL$, $\O$, $\U$, $\Sp$ be the symmetric, general linear, orthogonal (over $\R$), unitary (over $\C$), and symplectic group families. Then $\cc{TI}_{\GL} = \cc{TI}$ by definition, and $\cc{TI}_{\S}=\cc{GI}$, as asking if two $3$-tensors are the same up to permuting the coordinates is just the colored $3$-partite 3-uniform hypergraph isomorphism problem, a $\cc{GI}$-complete problem (by the methods of \cite{ZKT}). Then a special case of Theorem~\ref{thm:graph-iso} can be reformulated as $\cc{TI}_\S\subseteq \cc{TI}_\O\cap\cc{TI}_\U$, and special cases of Theorem~\ref{thm:toGL} can be reformulated as $\cc{TI}_\O, \cc{TI}_\Sp\subseteq \cc{TI}_\GL$. It may be interesting to investigate $\cc{TI}_\cG$ with $\cG$ being other subgroups of $\GL$, such as special linear, affine, and Borel or parabolic subgroups.

\paragraph{Open questions.} With this notation in hand, we highlight the following questions left open by our work:

\begin{question}
	Which, if any, of $\cc{TI}_{\O}, \cc{TI}_{\U}, \cc{TI}_{\Sp}$ are equal to $\cc{TI}$? 
\end{question}

As a warm-up in this direction, one may ask which of these classes is not only $\cc{GI}$-hard, but contains \CodeEqlong (permutational or monomial).

We suspect that $\cc{GI} \subseteq \cc{TI}_{\Sp} \cap \cc{TI}_{\SL}$ as well, for the following reason. Although the symplectic groups $\Sp_n$ and the special linear groups $\SL_n$ do not contain the symmetric group $S_n$ given by $n \times n$ permutation matrices, they do contain isomorphic copies of $S_{n'}$ for $n' \geq \Omega(n)$. In particular, $\Sp_{2n}$ contains $S_n$ as the subgroup $ \{ A \oplus A^T : A \in S_n\}$, and $\SL_n \cap S_n = A_n$ (and contains an isomorphic copy of $S_{n-2}$, where even $\pi \in S_{n-2}$ get embedded as $P_{\pi} \oplus I_2$ and odd $\pi$ get embedded as $P_{\pi} \oplus \tau$, where $\tau = \begin{bmatrix} 0 & 1 \\ 1 & 0 \end{bmatrix}$).  

\begin{question}
	Is $\cc{TI}_{\SL}$ contained in $\cc{TI}$? Are they equal?
\end{question}

\begin{question}
	Is $\cc{TI}_\U \subseteq \cc{TI}$? And the same question for unitary versus general linear group actions over finite fields.
\end{question}

\begin{question}
	What is the complexity of various problems in $\cc{TI}$ when restricted from $\GL$ to other form-preserving groups? A notable family of such groups is the mixed orthogonal groups $\O(p,q)$, defined over $\R$ by preserving a real symmetric form of signature $(p,q)$. But more generally, what about form-preserving groups for forms that are neither symmetric nor skew-symmetric?
\end{question}

\paragraph{Paper organisation.} After presenting some preliminaries in Section~\ref{sec:prel}, we prove the main results: Theorem~\ref{thm:graph-iso} in Section~\ref{sec:graph-iso}, Theorem~\ref{thm:toGL} in Section~\ref{sec:toGL}, Theorem~\ref{thm:5actions} in Section~\ref{sec:5actions}, and Theorem~\ref{thm:dto3} in Section~\ref{sec:dto3}.

\section{Preliminaries}\label{sec:prel}
\paragraph{Fields.} All our reductions are constant-free $p$-projections (that is, the only constants they use other than copying the ones already present in the input are $\{0,1,-1\}$). When the fields are representable on a Turing machine, our reductions are logspace computable. For arbitrary fields, the reductions are in logspace in the Blum--Shub--Smale model over the corresponding field.

\paragraph{Linear algebra.} All vector spaces in this
article are finite dimensional. Let $V$ be a vector space over a field $\F$. The
dual of $V$, $V^*$, consists of all linear or anti-linear forms over $\F$. In
this case when anti-linear is considered, $\F$ is a
quadratic extension of a subfield $\K$, there is thus an automorphism $\alpha \in \aut_\K(\F)$ of order two, and anti-linear means $f(\lambda v) = \alpha(\lambda) f(v)$. An example is $\F=\C$ and $\K=\R$, and $\alpha$=complex conjugation. Whether $V^*$ denotes linear or antilinear maps should be evident from context.

\paragraph{Some subgroups of general linear groups.}
Let $V$ be
a vector space over a field $\F$. Let $\GL(V)$ be
the general linear group over $V$, which consists of all invertible linear maps on
$V$. Let $\phi:V\times V\to \F$ be a bilinear or
sesquilinear form on $V$. In the case when $\phi$ is sesquilinear, $\F$ is a
quadratic extension of a subfield $\K$; sesquilinear means that it is linear in one argument and anti-linear in the other. Then $\GL(V)$ acts on $\phi$ naturally, by
$M\in\GL(V)$
sends $\phi$ to $\phi\circ M$, defined as $(\phi\circ M)(v, v')=\phi(M(v), M(v'))$.
The subgroup of $\GL(V)$ that preserves $\phi$ is
denoted as $\cG(V, \phi):=\{M\in\GL(V) \mid \phi\circ M=\phi\}$.

It is well-known that some classical groups arise as $\cG(V, \phi)$.
\begin{enumerate}
	\item Let $\F=\mathbb{C}$. Let $\phi$ be the sesquilinear form on $V=\mathbb{C}^n$
	defined as
	$\phi(u, v)=\sum_{i\in[n]}u_i^*v_i$, where $u_i^*$ is the complex conjugate of
	$u_i$. Then $\cG(V, \phi)$ is the unitary group $\mathrm{U}(n, \mathbb{C})$.
	\item Let $\F=\mathbb{R}$. Let $\phi$ be the symmetric bilinear form on
	$V=\mathbb{R}^n$ defined as
	$\phi(u, v)=\sum_{i\in[n]}u_i v_i$. Then $\cG(V, \phi)$ is the orthogonal group $\mathrm{O}(n, \mathbb{R})$.
	\item Let $\phi$ be the skew-symmetric bilinear form on $V=\mathbb{F}^{2n}$, defined as
	$\phi(u, v)=\sum_{i\in[n]}(u_i v_{2n-i+1}-u_{n+i} v_{n-i+1})$. Then $\cG(V, \phi)$ is the symplectic group $\mathrm{Sp}(2n, \mathbb{F})$.
\end{enumerate}

Depending on the underlying fields, orthogonal groups may indicate some families of groups preserving different (non-congruent) symmetric forms. In this paper we always use orthogonal groups and unitary groups w.r.t. the standard bilinear or sesquilinear form as defined above. 

\paragraph{Matrices.} Let $\M(l\times m, \F)$ be the linear space of $l\times m$ matrices over $\F$, and $\M(n,\F):=\M(n\times n,\F)$. Given $A\in\M(l\times m, \F)$, denote by $A^t$ the transpose of $A$. Given $A\in\GL(n, \F)$, denote by $A^{-1}$ the inverse of $A$ and by $A^{-t}$ the inverse transpose of $A$. 

We use $I_n$ to denote the $n\times n$ \emph{identity matrix}, and if it is clear from the context, we may drop the subscript $n$. For $(i, j)\in[n]\times [n]$, let $\E_{i,j}\in \M(n, \F)$ be the \emph{elementary matrix} where the $(i, j)$th entry is $1$, and the remaining entries are $0$. For $i \neq j$, the matrix $\E_{i,j}-\E_{j,i}$ is called an \emph{elementary alternating matrix}.

\paragraph{3-way arrays and some group actions on them.} 
Let $\T(\ell\times m\times n, \F)$ be the linear space of $\ell\times m\times n$ 3-way arrays over $\F$. 
Given $\tA\in \T(\ell\times m\times n, \F)$, the $(i,j,k)$th entry of $\tA$ is 
denoted as $\tA(i,j,k)\in \F$. We can slice $\tA$ along one direction and obtain several matrices, which are called slices. For example, slicing along the third coordinate, we obtain the \emph{frontal} slices, namely $n$ matrices $A_1, \dots, A_n\in \M(l\times m, \F)$, where $A_k(i,j)=\tA(i,j,k)$. Similarly, we also obtain the \emph{horizontal} slices by slicing along the first coordinate, and the \emph{lateral} slices by slicing along the second coordinate. 

A $3$-way array allows for group actions in three directions. 
Given $P\in \M(\ell, \F)$ and $Q\in \M(m, \F)$, let $P\tA Q$ be the $\ell \times m\times n$ $3$-way array whose $k$th frontal slice is $P A_k Q$. For $R=(r_{i,j})\in \M(n, \F)$, let $\tA^R$ be the $\ell\times m\times n$ $3$-way array whose $k$th frontal slice is $\sum_{k'\in[n]}r_{k',k}A_{k'}$. 

\paragraph{Tensors.} Let $V_1, \dots, V_c$ be vector spaces over $\F$. Let $a_i,
b_i, i\in[c]$ be non-negative integers, such that for each $i$, $a_i + b_i > 0$. A tensor $T$ of type $(a_1, b_1; a_2, b_2;
\dots;
a_c, b_c)$ supported by $(V_1, \dots, V_c)$ is an element in $V_1^{\otimes
	a_1}\otimes V_1^{*\otimes b_1} \otimes
V_2^{\otimes a_2}\otimes V_2^{*\otimes b_2}\otimes\dots\otimes V_c^{\otimes
	a_c}\otimes V_c^{*\otimes b_c}$. We say that $V_i$'s are the supporting vector
spaces of $T$, and $a_i$ (resp. $b_i$) is the multiplicity of $T$ at $V_i$ (resp.
$V_i^*$). (By convention $V^{\otimes 0} := \F$; note that $U \otimes \F \cong U$, since our tensor products are over $\F$.)

The order of $T$ is $\sum_{i\in[c]}(a_i+b_i)$. We say that $T$ is \emph{plain},
if $a_1=\dots=a_c=1$ and $b_1=\dots=b_c=0$.
The group $\GL(V_1)\times\dots\times\GL(V_c)$ acts naturally on the space $V_1^{\otimes a_1}\otimes V_1^{*\otimes b_1} \otimes
V_2^{\otimes a_2}\otimes V_2^{*\otimes b_2}\otimes\dots\otimes V_c^{\otimes a_c}\otimes V_c^{*\otimes b_c}$. Two tensors in this space are isomorphic if they are in the same orbit under this group action.

\paragraph{From tensors to multiway arrays.} For $i\in[c]$, let $V_i$ be a dimension-$d_i$ vector
space over $\F$. Let $T$ be a tensor in $V_1^{\otimes
	a_1}\otimes V_1^{*\otimes b_1} \otimes
V_2^{\otimes a_2}\otimes V_2^{*\otimes b_2}\otimes\dots\otimes V_c^{\otimes
	a_c}\otimes V_c^{*\otimes b_c}$. After fixing the basis of each $V_i$,
$T$ can be represented as a multiway array 
$R_T\in\T(d_1^{\times (a_1+b_1)} \times \dots \times d_c^{\times (a_c+b_c)})$ and the elements in $\GL(V_i)\cong \GL(d_i, \F)$ can be represented as invertible $d_i \times d_i$ matrices. The action of $(A_1, \dots, A_c)$ on $R_T$ can be explicitly written following Definition~\ref{def:tensor}, using $A_i$ for $a_i$ directions and $A_i^{-t}$ for $b_i$ directions. 

\section{Proof of Theorem~\ref{thm:graph-iso}}\label{sec:graph-iso}

Our goal is to reduce the directed graph isomorphism problem \DGI, which is $\cc{GI}$-complete \cite{KST93}, to the isomorphism problem of $U\otimes U\otimes V$ under $\cG(U)\times\cG(V)$ where $\cG=\{\cG_n\}$ is a family of subgroups of the general linear groups that contains the symmetric groups. Note that $\S_n$ is naturally embedded in $\GL(n, \F)$ by taking the matrix representation of permutations. In this section, we also view $\cG_n$ as a matrix group.

Recall that two directed graphs $G=([n], E)$ and $H=([n], F)$ are isomorphic, if there exists a bijective map $f:[n]\to [n]$, such that $(i, j)\in E$ if and only if $(f(i), f(j))\in F$.

We rephrase~\cite[Proposition 6.1, Proposition 6.2]{LQWWZ22} to adapt them to our context in the following proposition, which would conclude our proof of Theorem~\ref{thm:graph-iso}.

\begin{proposition}\label{thm: gi to ti}
	Given two directed graphs $G=([n], E)$ and $H=([n], F)$, and a group family $\cG=\{\cG_n\}$ satisfying that $\S_n\leq\cG_n\leq\GL(n,\F)$. We can efficiently construct two $3$-tensors $\tS_G$ and $\tS_H$ associated with $G$ and $H$, respectively, such that $\tS_G$ is isomorphic to $\tS_H$ in $U\otimes U\otimes W$ under the action of $\cG(U)\times \cG(W)$ if and only if $G$ is isomorphic to $H$.
\end{proposition}

\begin{proof}
	\paragraph{The construction. }
	For directed graph $G=([n],E)$, we construct the associated 3-way array $\tS_G\in \T(n\times n\times |E|, \F)$ by setting its frontal slices as $(\E_{i, j} \mid (i, j)\in E)$, where edges in $E$ are ordered lexicographically. We also construct $\tS_H\in \T(n\times n\times |F|, \F)$ associated with $H=([n],F)$ in the same way. Let $m = |E| = |F|$. We will show that $G \cong H$ as graphs iff $\tS_G$ and $\tS_H$ are in the same orbit of $\cG(U) \times \cG(V)$ acting on $U \otimes U \otimes V$, where $U = \F^n, V= \F^{m}$.
	
	\paragraph{The if direction.}
	Let $\sigma\in\S_n$ be an isomorphism from $G$ to $H$, and let $\tau\in\S_m$ be the induced permutation of edges. Then the permutation matrices corresponding to $\sigma$ and $\tau$ yield an isomorphism from $\tS_G$ to $\tS_H$. Note that here we need to use the condition that $\cG$ contains symmetric groups.
	
	\paragraph{The only if direction.}
	Let $T\in\cG_n\leq\GL(n,\F)$ and $R\in\cG_m\leq\GL(m,\F)$ such that $T^t\tS_GT= \tS_H^R$. Denote by $A_r$ the $r$th frontal slice of $\tS_H^R$. Let $t_{i,j}\in\F$ be the $(i,j)$th entry of $T$. Then for each $(i,j)\in E$, let $r\in[m]$ be the index corresponding to the edge $(i,j) \in E$. Then we have:
	\[
	T^t\E_{i,j}T
	=\begin{bmatrix}t_{1,i}\\\vdots\\t_{n,i}\end{bmatrix}\begin{bmatrix}t_{1,j}&\cdots&t_{n,j}\end{bmatrix}=[t_{k,i}t_{\ell,j}]_{k,\ell\in[n]}= A_r.
	\]
	Note that for each $r\in[m]$, $A_r$ is a linear combination of the frontal slices of $\tS_H$. Since the slices of $\tS_H$ are of the form $E_{i,j}$, and these are linearly independent, it follows that if the $(k,\ell)$th entry $t_{k,i}t_{\ell,j}$ of $A_r$ is non-zero for some $r\in[m]$, then $\E_{k,\ell}$ must be present with nonzero coefficient in this linear combination, and therefore we must have $(k,\ell) \in F$. 
	
	Next, as $T$ is invertible, there exists a permutation $\sigma\in\S_n$ such that $t_{\sigma(i),i}\neq 0$ for all $i\in[n]$ (for otherwise, considering the expression of $\det T$ as a sum over permutations, we would get $\det T =0$). 
	In particular, $t_{\sigma(i),i}t_{\sigma(j),j}\neq 0$ for 
	any $i,j\in[n]$. Combining with the previous paragraph, we get that for $(i,j)\in E$, 
	$(\sigma(i),\sigma(j))\in F$. In other words, $\sigma$ is an injective map from 
	vertices of $G$ to vertices of $H$ which preserves arcs. Finally, the invertibility of $T$ and $R$ ensures that the frontal slices of $\tS_G$ and $\tS_H$ have the same number, which means $G$ and $H$ have the same edge size, and hence that $\sigma$ in fact induces a bijection on arcs. This shows that $G$ is isomorphic to $H$, as claimed.
\end{proof}

\section{Proof of Theorem~\ref{thm:toGL}}\label{sec:toGL}

Let $\{\phi_n:\F^n\times \F^n\to \F\}$ be a family of bilinear forms.  Suppose $\cG_n\leq\GL(n, \F)$ preserves $\phi_n$. For convenience, we only consider one action in Definition~\ref{def:five-action}, and the reader will see that the other four actions follow essentially the same line of arguments. Our goal is to decide whether $\tA, \tB\in \T(l\times m\times n, \F)$ are in the same orbit under the action of $\cG_l\times\cG_m\times\cG_n$. We would like to reduce to the isomorphism problem of $U\otimes V\otimes W$ under $\GL(U)\times\GL(V)\otimes\GL(W)$ where the dimensions of $U, V, W$ are polynomial in $l, m, n$. 

The key to the reduction is \cite[Theorem 1.1]{FGS19}. For this, we need the tensor system notion in \cite{FGS19}.
This notion is also related to tensor networks, and we refer the reader to
\cite{FGS19} for further references.

\paragraph{Tensor systems and \cite[Theorem 1.1]{FGS19}.} Let $V=\{V_1, \dots, V_c\}$ be a set of vector spaces over a field $\F$. Let
$T=\{T_1, \dots, T_n\}$ be a set of tensors, such that $T_i$ is supported by a
subset of $V_j$'s in $V$.

The types of $T_i$'s can be recorded by a bipartite
graph (with directed, possibly parallel arcs) as follows. Let $B_T=(T\cup V, E)$ be a
bipartite graph, where $E$ is a multiset whose elements are from $V\times T$ and
$T\times V$. The arcs in $E$ are as follows. Suppose the multiplicity of $T_i$ at
$V_j$ (resp. $V_j^*$) is $a_{i,j}$ (resp. $b_{i,j}$). Then the multiplicity of
$(V_j, T_i)$ (resp. $(T_i, V_j)$) in $E$ is $a_{i, j}$ (resp. $b_{i,
	j}$).
For an example, see Figure~\ref{fig:system}. 
\begin{figure}[H]
	\centering
	\begin{tikzpicture}[thick,scale=0.8, every node/.style={scale=0.9}]
	\node (T3) at (0,0) {$T_3$};
	\node (V2) at (4,1) {$V_2$};
	\node (T2) at (0,2) {$T_2$};
	\node (V1) at (4,3) {$V_1$};
	\node (T1) at (0,4) {$T_1$};
	\draw[-latex] (T1) to (V1);
	\draw[-latex] (V1) to (T2);
	\draw[-latex] (V2) to (T1);
	\draw[-latex,bend right] (T2) to (V2);
	\draw[-latex] (T2) to (V2);
	\draw[-latex] (V1) to (T3);
	\draw[-latex,bend left] (V1) to (T3);
	\draw[-latex,bend right] (T3) to (V2);
	\end{tikzpicture}
	\caption{ \label{fig:system} The bipartite graph encoding a system of three tensors over two $\F$-vector spaces $V_1, V_2$: $T_{1} \in V_{1}^{*}\otimes V_{2}$, $T_{2} \in V_{1} \otimes V_{2}^{*} \otimes V_{2}^{*}$, and $T_{3} \in V_{1} \otimes V_{1} \otimes V_{2}^{*}$.}
\end{figure}
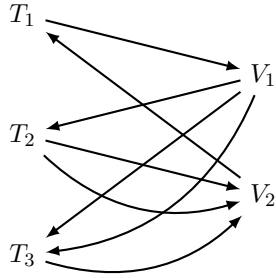

Let $S=(S_1, \dots, S_c)$ and $T=(T_1, \dots, T_c)$ be two tensor systems of the
same type. That is, their underlying bipartite graphs are the same up to renaming
$S_i$ with $T_i$. We say that $S$ and $T$ are isomorphic if and only if there
exists $A=(A_1, \dots, A_c)\in\GL(V_1)\times\dots\times\GL(V_c)$ such that (the
relevant components of) $A$ sends $S_i$ to $T_i$ for every $i\in[c]$.

\begin{theorem}[{Rephrase of \cite[Theorem 1.1]{FGS19}}]\label{thm:FGS_parameter}
	Let $S=\{S_1, \dots, S_c\}$ and $T=\{T_1, \dots, T_c\}$ be two tensor systems
	supported by $\{V_1, \dots, V_m\}$, where each $S_i$ and $T_i$ is of order $\leq 3$. 
	Then there exists an algorithm $A$ that takes
	$S$
	and $T$ and outputs two plain 3-tensors $A(S)$ and $A(T)$ supported by vector
	spaces $\{U, V,
	W\}$, such that $S$ and $T$ are isomorphic as tensor systems if and
	only if $A(S)$ and $A(T)$ are isomorphic. The algorithm runs in time polynomial in
	the dimension of $U,V,W$, and this maximum dimension is at most {$\poly(\sum_{i\in[m]}\dim(V_i),2^{\poly(c)})$}.
\end{theorem}

Given Theorem~\ref{thm:FGS_parameter}, we now proceed to the proof of Theorem~\ref{thm:toGL}. 

\begin{proof}[Proof of Theorem~\ref{thm:toGL}] Let $\Phi_l\in\M(l, \F)$ be the matrix representation of the bilinear form $\phi_\ell$. Then construct $S=(\tA, \Phi_l, \Phi_m, \Phi_n)$ and $T=(\tB, \Phi_l, \Phi_m, \Phi_n)$, viewed as tensor systems as follows. Let $U'=\F^l$, $V'=\F^m$, and $W'=\F^n$. Then $\tA\in U'\otimes V'\otimes W'$, $\Phi_l\in U'\otimes U'$, $\Phi_m\in V'\otimes V'$, and $\Phi_n\in W'\otimes W'$. It is clear that $\tA$ and $\tB$ are $\cG$-isomorphic if and only if the two tensor systems $S$ and $T$ are $\GL$-isomorphic. 
	
	Every tensor in the above tensor systems is of order $\leq 3$, and each has only $c=4=O(1)$ components, so we can apply Theorem~\ref{thm:toGL} to obtain two plain tensors $r(S)$ and $r(T)$ in $U\otimes V\otimes W$, where $\dim(U)$, $\dim(V)$ and $\dim(W)$ are at most by $\poly(l+n+m)$. Furthermore, $S$ and $T$ are isomorphic as tensor systems if and only if $r(S)$ and $r(T)$ are isomorphic as plain tensors. This concludes the proof. 
\end{proof}

\section{Proof of Theorem~\ref{thm:5actions}}\label{sec:5actions}

Recall that we need to show the polynomial-time equivalence between the isomorphism problems of $U\otimes V\otimes W$, $U\otimes U\otimes V$, $U\otimes U^*\otimes V$, $U\otimes U\otimes U$, and $U\otimes U\otimes U^*$ under orthogonal and unitary groups. We present the proofs for unitary groups, and the proofs for orthogonal groups follow the same line. 

The equivalences for $\GL$ were proved in \cite{FGS19,GQ21}. We follow their proof strategies, but as mentioned in Section~\ref{sec:overview}, certain technical difficulties need to be dealt with. 

In Section~\ref{subsec:FGS}, we reduce $U\otimes U\otimes V$, $U\otimes U^*\otimes V$, $U\otimes U\otimes U$, and $U\otimes U\otimes U^*$ to $U\otimes V\otimes W$. This is done through the tensor system framework with the adaptation to unitary isomorphism.

In Section~\ref{subsec:UVWtoUUW}, we reduce $U\otimes V\otimes W$ to $U\otimes U\otimes W$. This requires a careful check due to the introduction of the gadget. 

In Section~\ref{subsec:UVWtoUUstarW} we reduce 
$U\otimes V\otimes W$ to $U\otimes U^*\otimes W$. This requires the Singular Value Theorem as a new ingredient.

In Section~\ref{subsec:UUWtoUUU}, we reduce $U\otimes U\otimes W$ to $U\otimes U\otimes U^*$ and $U\otimes U\otimes U$.

\subsection{Reduction to plain \algprobm{Unitary $3$-Tensor Isomorphism}}\label{subsec:FGS}

In this section, we will reduce unitary isomorphism problems of $U\otimes U\otimes V$, $U\otimes U^*\otimes V$, $U\otimes U\otimes U$, and $U\otimes U\otimes U^*$ to $U\otimes V\otimes W$ with a polynomial dimension blow-up. This requires the following unitary version of Theorem~\ref{thm:FGS_parameter}, so we begin with its proof.

\begin{theorem}[{Unitary version of \cite[Theorem 1.1]{FGS19}}]\label{thm:FGS_unitary}
	Let $S=\{S_1, \dots, S_c\}$ and $T=\{T_1, \dots, T_c\}$ be two tensor systems
	supported by $\{V_1, \dots, V_m\}$, where every $S_i$ and $T_i$ is of order $\leq 3$. 
	Then there exists an algorithm $r$ that takes
	$S$
	and $T$ and outputs two 3-tensors $r(S)$ and $r(T)$ supported by vector
	spaces $\{U, V,
	W\}$, such that $S$ and $T$ are isomorphic as tensor systems under $\U(V_1)\times\dots\times\U(V_m)$ if and
	only if $r(S)$ and $r(T)$ are isomorphic under $\U(U)\times\U(V)\times\U(W)$. The algorithm $r$ runs in time polynomial in
	the maximum dimension over $U, V, W$, and this maximum dimension is upper bounded
	by {$\poly(\sum_{i\in[m]}\dim(V_i),2^{\poly(c)})$}.
\end{theorem}

This follows the same proof as \cite[Theorem 1.1]{FGS19}, outlined in Appendix~\ref{app:proof}, with one change, based on the following result. 

We say that  two matrix tuples $(C_1, \dots, C_m)\in \M(l\times n, \F)^m$ and $(D_1, \dots, D_m)\in\M(l\times n, \F)^m$ are unitarily equivalent, if there exist unitary matrices $L\in\U(l, \F)$ and $R\in\U(n, \F)$, such that for any $i\in[m]$, $LC_iR=D_i$. 
\begin{theorem}[{Sergeichuk \cite[Theorem 3.1]{Ser98}}]\label{thm:unitary}
	Let $\vC=(C_1, \dots, C_m)\in\M(l\times n, \F)$. Suppose $\vC$ is unitarily equivalent to $\vD=(D_1, \dots, D_m)$, such that each $D_i$ is block-diagonal with $k$ blocks, with the $j$th block of size $d_j\times d_j$. Furthermore, let $\vD_j=(D_{1, j}, \dots, D_{m, j})$ be the $m$-tuple of $d_j \times d_j$ matrices consisting of the $j$th block from each $D_i$, and suppose $\vD_j$ is not unitarily equivalent to a block-diagonal tuple. Then the isomorphism types of $\vD_i$'s and the multiplicities of each isomorphism type are uniquely determined by $\vC$, that is, they are the same regardless of the choice of decomposition. 
\end{theorem}

From the above theorem, the following corollary is immediate:

\begin{corollary} \label{cor:unitary}
	If $\left(\begin{bmatrix} A_1 & 0 \\0 & B_1\end{bmatrix}, \dots, \begin{bmatrix} A_m & 0 \\0 & B_m\end{bmatrix}\right)$ and $\left(\begin{bmatrix} A_1 & 0 \\0 & C_1\end{bmatrix}, \dots, \begin{bmatrix} A_m & 0 \\0 & C_m\end{bmatrix}\right)$ are unitarily equivalent, then $(B_1, \dots, B_m)$ and $(C_1, \dots, C_m)$ are unitarily equivalent. 
\end{corollary}

\begin{proof}[Proof of Theorem~\ref{thm:FGS_unitary}]
	With Corollary~\ref{cor:unitary}, the proof of \cite[Theorem 1.1]{FGS19} goes through for this unitary setting, by replacing the use of the Krull--Schmidt theorem for quiver representations (\cite[pp. 20]{FGS19}) with Theorem~\ref{thm:unitary}.
	
	The case of orthogonal groups follows similarly by using \cite[Theorem 4.1]{Ser98} instead.
\end{proof}

We utilize the tensor system to construct reductions to plain 3-tensor unitary isomorphism, and then prove their correctness by Theorem \ref{thm:FGS_unitary}.
\begin{proposition}
	\label{cor: down-up}
	The unitary isomorphism problems on $V\otimes V\otimes W,V\otimes V^{*}\otimes W,V\otimes V\otimes V$ and $V\otimes V\otimes V^{*}$ are polynomial-time reducible to \algprobm{Unitary 3-Tensor Isomorphism} on $U'\otimes V'\otimes W'$ where $\dim(U'),\dim(V')$ and $\dim(W')$ are at most polynomial in $\dim(V)$ and $\dim(W)$.
\end{proposition}
\begin{proof}
	The reduction is based on the observation that tensor systems can encode these isomorphism problems. For example, for $\tA \in V\otimes V\otimes W$, we can construct a tensor system consisting of one tensor $\tA$ and two vector spaces $\{V, W\}$, with two arcs from $V$ to $\tA$, and one arc from $W$ to $\tA$. Starting from two tensors $\tA_1, \tA_2\in V\otimes V\otimes W$, we consider the corresponding tensor systems, and ask for unitary isomorphism of these tensor systems. 
	Then by Theorem \ref{thm:FGS_unitary}, they can be reduced to the plain 3-tensor unitary isomorphism in time $\poly(\dim(V), \dim(W))$, as these are tensor systems with only 1 tensor each.
	It can be seen that this works for $V\otimes V^*\otimes W$, $V\otimes V\otimes V$, and $V\otimes V\otimes V^*$. This concludes the proof. 
\end{proof}

\subsection{Reduction from \algprobm{Unitary $3$-TI} to \algprobm{Bilinear Form Unitary Psuedoisometry} ($V\otimes V\otimes W$)}\label{subsec:UVWtoUUW}

We mainly follow the construction in \cite{GQ21} to show that there is a reduction from \algprobm{Unitary $3$-Tensor Isomorphism} ($U\otimes V\otimes W$) to \algprobm{Bilinear Form Unitary Pseudoisometry} ($V'\otimes V'\otimes W'$). In addition, we prove that the reduction from \cite{GQ21} preserves the unitary property in both directions.

\begin{proposition}
	\label{lem: act1 to act2}
	Given two 3-tensors $\tA,\tB\in U \otimes V \otimes W$, where $\dim(U)=l\leq \dim(V)=m$ and $\dim(W)=n$. There is a reduction $r: U \otimes V \otimes W \to V' \otimes V' \otimes W'$ with $\dim(V')=l+5m+3$ and $dim(W')=n+l(m+1)+m(3m+2)$ such that $\tA$ and $\tB$ are unitarily isomorphic if and only if $r(\tA)$ and $r(\tB)$ are unitarily isomorphic, where frontal slices of $r(\tA)$ and $r(\tB)$ are skew-symmetric matrices.
\end{proposition}
\begin{proof}
	\paragraph{The reduction.} We use the gadget in \cite{FGS19} and \cite{GQ21} to present this reduction. Here we use matrix format to illustrate our construction, and the picture of this construction is shown in Figure \ref{fig:3-tensor_isometry}. Denote the $i$th frontal slice of $\tA$ by $A_i\in \mathrm{M}(l\times m, \C)$, where $i\in\left[ n \right]$. Let the $i$th frontal slice of $r(\tA)$ be $\hat{A}_i \in \mathrm{M}(l+5m+3, \C)$, where $i\in\left[ n+l(m+1)+m(3m+2) \right]$. Then $\hat{A}_i$ is constructed as follows:
	\begin{itemize}
		\item For $i\in\left[ n \right]$, $\hat{A}_i$ is of the form $\begin{bmatrix}
		\mathbf{0} & A_i        & \mathbf{0} \\
		-A_i^t     & \mathbf{0} & \mathbf{0} \\
		\mathbf{0} & \mathbf{0} & \mathbf{0}
		\end{bmatrix}$.
		\item For $i\in\left[ n+1,n+l(m+1) \right]$, let $\hat{A}_i$ be the elementary alternating matrix $\E_{s,l+m+t}-\E_{l+m+t,s}$, where $s = \lceil (i-n)/(m+1)\rceil$ and $t = i-n-(s-1)(m+1)$.
		\item For $i\in\left[ n+l(m+1),n+l(m+1)+m(3m+2) \right]$, let $\hat{A}_i$ be the elementary alternating matrix $\E_{l+s,l+m+m+1+t}-\E_{l+m+m+1+t,l+s}$, where $s = \lceil (i-n-l(m+1))/(3m+2)\rceil$ and $t = i-n-l(m+1)-(s-1)(3m+2)$.
	\end{itemize}
	
	Denote lateral slices of $r(\tA)$ by $L_i$, where $i\in[l+5m+3]$. Then we check the ranks of these lateral slices:
	\begin{itemize}
		\item For the first $l$ slices, the lateral slice $L_i$ is a block matrix with two non-zero blocks. One block is $-I_{m+1}$, and another block of size $m\times n$ is the transpose of the $i$th horizontal slice of $-\tA$. Thus, $m+1\leq\rank(L_i)\leq 2m+1$.
		\item For the following $m$ slices, $L_i$ is a block matrix with two non-zero blocks. One block is $-I_{3m+2}$ and the other one is the $(i-n)$th lateral slice of $\tA$ with size $l\times n$. Therefore, $3m+2\leq\rank(L_i)\leq 3m+2+l\leq 4m+2$.
		\item For the next $m+1$ slices, $L_i$ has a block $I_{l}$ after rearranging the columns, so $\rank(L_i)= l\leq m$.
		\item For the last $3m+2$ slices, similarly, $L_i$ has a block $I_m$ after rearranging the columns, so $\rank(L_i)= m$.
	\end{itemize}
	
	Now we consider the ranks of linear combinations of the above slices. There are four observations that help prove the correctness of the reduction:
	\begin{itemize}
		\item If the combination contains $L_i$ for $1\leq i\leq l$, since the resulting matrix has at least one identity matrix $I_{m+1}$ in the $(l+m+1)$th row to $(l+2m+1)$th row, it has the rank at least $m+1$.
		\item If the combination doesn't contain $L_i$ for $l+1\leq i\leq l+m+1$, the resulting matrix has rank at most $3m+1$, because there are at most $l+5m+3-3m-2\leq 3m+1$ non-zero rows.
		\item If the combination involves $L_i$ for $l+1\leq i\leq l+m+1$, the resulting matrix has rank at least $3m+2$, because there is at least one identity matrix $I_{3m+2}$ in the last $3m+2$ rows.
		\item If the combination involves $L_i$ for $1\leq i\leq l$ and $L_i$ for $l+1\leq i\leq l+m+1$, the resulting matrix has rank at least $4m+3$, because there are at least one identity matrix $I_{3m+2}$ in the last $3m+2$ rows and one identity matrix $I_{m+1}$ in the $(l+m+1)$th row to $(l+2m+1)$th row.
	\end{itemize}
	
	\paragraph{The if direction.} Assume there are $P\in \mathrm{U}(l+5m+3,\C)$ and $Q\in\mathrm{U}(n+l(m+1)+m(3m+2),\C)$ such that $P^t r(\tA)P=r(\tB)^Q$. Then we write $P$ as $P=\begin{bmatrix}
	P_{1,1} & P_{1,2} & P_{1,3} \\
	P_{2,1} & P_{2,2} & P_{2,3} \\
	P_{3,1} & P_{3,2} & P_{3,3}
	\end{bmatrix}$, where $P_{1,1}\in \mathrm{M}(l, \C)$, $P_{2,2}\in \mathrm{M}(m, \C)$ and $P_{3,3}\in \mathrm{M}(4m+3, \C)$. By ranks of lateral slices of $r(\tB)$ and the above observations, it's easy to have that $P_{2,1}=\mathbf{0},P_{1,2}=\mathbf{0},P_{1,3}=\mathbf{0}$ and $P_{2,3}=\mathbf{0}$. Therefore, $P$ is of the form $\begin{bmatrix}
	P_{1,1}    & \mathbf{0} & \mathbf{0} \\
	\mathbf{0} & P_{2,2}    & \mathbf{0} \\
	P_{3,1}    & P_{3,2}    & P_{3,3}
	\end{bmatrix}$. As $P$ is a block-lower-trianglular unitary matrix, $P_{1,1},P_{2,2}$ and $P_{3,3}$ are unitary matrices. Since the aim is to check if $\tA$ and $\tB$ are isomorphic, we only consider the first $n$ frontal slices of $r(\tA)$ and $r(\tB)$, which contains $\tA$ and $\tB$ respectively. After applying $P$ on lateral slices and horizontal slices of $r(\tA)$, we have the first $n$ frontal slices as follows:
	\begin{align*}
	\begin{bmatrix}
	P_{1,1}^t  & \mathbf{0} & P_{3,1}^t \\
	\mathbf{0} & P_{2,2}^t  & P_{3,2}^t \\
	\mathbf{0} & \mathbf{0} & P_{3,3}^t
	\end{bmatrix}\begin{bmatrix}
	\mathbf{0} & A_i        & \mathbf{0} \\
	-A_i^t     & \mathbf{0} & \mathbf{0} \\
	\mathbf{0} & \mathbf{0} & \mathbf{0}
	\end{bmatrix}\begin{bmatrix}
	P_{1,1}    & \mathbf{0} & \mathbf{0} \\
	\mathbf{0} & P_{2,2}    & \mathbf{0} \\
	P_{3,1}    & P_{3,2}    & P_{3,3}
	\end{bmatrix}=\begin{bmatrix}
	\mathbf{0}               & P_{1,1}^t A_i P_{2,2} & \mathbf{0} \\
	-P_{2,2}^t A_i^t P_{1,1} & \mathbf{0}            & \mathbf{0} \\
	\mathbf{0}               & \mathbf{0}            & \mathbf{0}
	\end{bmatrix}.
	\end{align*}
	Then we apply the unitary matrix $Q$ on the frontal slices of $r(\tB)$, and have $P^t r(\tA)P=r(\tB)^Q$. Note that only the block $(1,2)$ and $(2,1)$ are non-zero blocks in the first $n$ slices of $r(\tB)$ and $P^t r(\tA)P$, so we have that only the first $n\times n$ submatrix $Q_{1,1}$ of $Q$ is non-zero in the first $n$ columns, which implies that $Q_{1,1}$ is unitary from the fact that $Q$ is unitary. Therefore, it is enough to give the isomorphism $P_{1,1}^t \tA P_{2,2}=\tB^{Q_{1,1}}$ where $P_{1,1}^t,P_{2,2}$ and $Q_{1,1}$ are unitary.
	
	\paragraph{The only if direction.} Assume $P\tA Q=\tB^R$ for some $P\in\mathrm{U}(l,\C), Q\in\mathrm{U}(m,\C)$ and $R\in\mathrm{U}(n,\C)$. We claim that there are two unitary matrices $\hat{P}=\diag(P,Q,S_1,S_2)\in\mathrm{U}(l+5m+3,\C)$ and $\hat{Q}=\diag(R,T_1,T_2)\in\mathrm{U}(n+l(m+1)+m(3m+2),\C)$
	such that $\hat{P}^t r(\tA)\hat{P}=r(\tB)^{\hat{Q}}$, where $S_1\in\mathrm{U}(m+1,\C),S_2\in\mathrm{U}(3m+2,\C),T_1\in\mathrm{U}(l(m+1),\C)$ and $T_2\in\mathrm{U}(m(3m+2),\C)$.
	
	Due to the fact that $P\tA Q=\tB^R$, it's straightforward to check the first $n$ frontal slices of $\hat{P}^t r(\tA)\hat{P}$ and $r(\tB)^{\hat{Q}}$ are equal. Then we consider the remaining gadget slices. Let $\overline{r(\tA)}$ and $\overline{r(\tB)}$ be tensors constructed by the $(m+1)$th frontal slice to $(m+l(m+1))$th frontal slice of $r(\tA)$ and $r(\tB)$, respectively. Consider $\overline{r(\tA)}$ and $\overline{r(\tB)}$ from the frontal view:
	\begin{align*}
	\begin{bmatrix}
	\mathtt{0}  & \mathtt{0} & \tE & \mathtt{0} \\
	\mathtt{0}  & \mathtt{0} & \mathtt{0} & \mathtt{0} \\
	-\tE & \mathtt{0} & \mathtt{0} & \mathtt{0} \\
	\mathtt{0}  & \mathtt{0} & \mathtt{0} & \mathtt{0}
	\end{bmatrix},
	\end{align*}
	where $\tE\in \T(l\times(m+1)\times l(m+1),\C)$.  Then we apply $\hat{P}$ on the lateral and horizontal slices of $\overline{r(\tA)}$,
	\begin{align*}
	\begin{bmatrix}
	P^t &     &       &       \\
	& Q^t &       &       \\
	&     & S_1^t &       \\
	&     &       & S_2^t
	\end{bmatrix}
	\begin{bmatrix}
	\mathbf{0} & \mathbf{0} & E_i        & \mathbf{0} \\
	\mathbf{0} & \mathbf{0} & \mathbf{0} & \mathbf{0} \\
	-E_i       & \mathbf{0} & \mathbf{0} & \mathbf{0} \\
	\mathbf{0} & \mathbf{0} & \mathbf{0} & \mathbf{0}
	\end{bmatrix}
	\begin{bmatrix}
	P &   &     &     \\
	& Q &     &     \\
	&   & S_1 &     \\
	&   &     & S_2
	\end{bmatrix}=
	\begin{bmatrix}
	\mathbf{0}   & \mathbf{0} & P^tE_iS_1  & \mathbf{0} \\
	\mathbf{0}   & \mathbf{0} & \mathbf{0} & \mathbf{0} \\
	-S_1^tE_i P & \mathbf{0} & \mathbf{0} & \mathbf{0} \\
	\mathbf{0}   & \mathbf{0} & \mathbf{0} & \mathbf{0}
	\end{bmatrix},
	\end{align*}
	where $E_i\in\mathrm{M}(l\times(m+1),\C)$. Observe that $P^t$ acts on the horizontal direction of $E$, so it requires designing proper $S_1$ and $T_1$ to remove the effect of $P$. Let the lateral slice of $\tE$ to be $L_i\in\mathrm{M}(l\times l(m+1),\C)$ where $i\in[m+1]$. Apply a proper permutation $\pi$ on the columns of $L_i$ and have the matrix $L_i'=L_iT_{\pi}=\begin{bmatrix}
	\mathbf{0}\ldots I_l\ldots\mathbf{0}
	\end{bmatrix}$ where $T_{\pi}\in\mathrm{M}(l(m+1),\C)$ is the permutation matrix and the $i$th block of $L_i'$ is the identity matrix $I_l\in\mathrm{M}(l,\C)$. After left multiplying $L_i'$ by $P^t$, we have $P^tL_i'=\begin{bmatrix}
	\mathbf{0}\ldots P^t\ldots\mathbf{0}
	\end{bmatrix}$. Now we define a diagonal matrix $T_1'$ as $\diag(P^t,\ldots,P^t)$, which gives us $P^tL_i'=L_i'T_1'\iff P^tL_i=L_iT_{\pi}T_1'T_{\pi}^t$. Then we set $S_1$ to be the identity matrix and $T_1$ to be $T_{\pi}T_1'T_{\pi}^t$, and it yields $P^t\tE S_1=\tE^{T_1}$, where $S_1$ and $T_1$ are unitary.
	
	It remains to check the last $m(3m+2)$ frontal slices, which uses the similar method as above, and this produces unitary matrix $S_2$ and $T_2$. Now we have the unitary matrix $S$ and $T$ as desired.
\end{proof}

\subsection{Reduction from \algprobm{Unitary $3$-Tensor Isomorphism} to \algprobm{Unitary Matrix Space Conjugacy} ($V\otimes V^{*}\otimes W$)}
\label{subsec:UVWtoUUstarW}

A 3-way array $\tA\in \T(l\times m\times n,\F)$ is \textit{non-degenerate} if along each direction, the slices are linearly independent. 

\begin{lemma}
	\label{lem: svd}
	For any 3-way array $\tA\in \T(l\times m\times n,\C)$, there are unitary matrices $T_1\in\mathrm{U}(l,\C),T_2\in\mathrm{U}(m,\C)$ and $T_3\in\mathrm{U}(n,\C)$ such that
	\begin{align*}
	(T_1\tA T_2)^{T_3}=\begin{bmatrix}
	\tilde{\tA} & \mathtt{0} \\
	\mathtt{0}           & \mathtt{0}
	\end{bmatrix},
	\end{align*}
	where $\tilde{\tA}$ is a non-degenerate array of size $l'\times m'\times n'$.
\end{lemma}
\begin{proof}
	First, we consider the horizontal slices of $\tA$. Let $(A_1,\ldots,A_n)$ be the corresponding matrix tuple of frontal slices of $\tA$. Then we construct the $l\times mn$ matrix
	\begin{align*}
	A'=\begin{bmatrix}
	A_1 & \ldots & A_n
	\end{bmatrix}.
	\end{align*}
	We denote the maximum number of linearly independent horizontal slices of $\tA$ by $l'$; it follows that the rank of $A'$ is $l'$. Applying a singular value decomposition on $A'$, we have
	\begin{align*}
	A' = U\Sigma V^*,
	\end{align*}
	where $U$ and $V$ are unitary matrices of size $l \times l$ and $mn \times mn$, respectively, and $\Sigma=\begin{bmatrix}
	\hat{\Sigma} \\
	\mathbf{0}
	\end{bmatrix}$ for a full-rank rectangular diagonal matrix $\hat{\Sigma}$ of size $l'\times mn$. Multiplying $A'$ by $T_1 = U^{-1}$, we have
	\begin{align*}
	T_1A' = \Sigma V^*,
	\end{align*}
	where the first $l'$ rows of $\Sigma V^*$ are linearly independent and the last $l-l'$ rows are zero. It follows that acting $T_1$ on the horizontal slices of $\tA$ sends $\tA$ to
	\begin{align*}
	T_1\tA = \begin{bmatrix}
	\hat{\tA} \\
	\mathtt{0}
	\end{bmatrix},
	\end{align*}
	where the horizontal slices of $\hat{\tA}\in\T(l'\times m\times n, \C)$ are linearly independent.
	
	We can similarly find unitary matrices $T_2, T_3$ for the other two directions.
\end{proof}

\begin{lemma}
	\label{obs: non-deg}
	Given two 3-tensors $\tA,\tB\in U\otimes V\otimes W$ where $l=\dim(U),m=\dim(V)$ and $n=\dim(W)$, there is a reduction $r$ such that $\tA$ and $\tB$ are unitarily isomorphic if and only if $r(\tA)$ and $r(\tB)$ are unitarily isomorphic, where $r(\tA)$ and $r(\tB)$ are non-degenerate.
\end{lemma}

We note that this reduction is one of the few in the paper that is explicitly \emph{not} a $p$-projection (similar to how the reduction of a matrix to row echelon form is not a $p$-projection).

\begin{proof}
	By Lemma \ref{lem: svd}, we can find unitary matrices $S_1\in\mathrm{U}(l,\C),S_2\in\mathrm{U}(m,\C)$ and $S_3\in\mathrm{U}(n,\C)$ to extract the $l'\times m'\times n'$ non-degenerate tensor $\tilde{\tA}$ of $\tA$. There are similar unitary matrices $T_1\in\mathrm{U}(l,\C),T_2\in\mathrm{U}(m,\C)$ and $T_3\in\mathrm{U}(n,\C)$ for $\tB$ as well. Then we claim $\tA$ and $\tB$ are unitarily isomorphic if and only if $r(\tA)=\tilde{\tA}$ and $r(\tB)=\tilde{\tB}$ are unitarily isomorphic.
	
	For the if direction, assume $\tilde{P}\tilde{\tA}\tilde{Q}=\tilde{\tB}^{\tilde{R}}$ where $\tilde{P}\in\mathrm{U}(l',\C),\tilde{Q}\in\mathrm{U}(m',\C)$ and $\tilde{R}\in\mathrm{U}(n',\C)$. It yields that $P'\tA'Q'=\tB'^{R'}$ where $\tA'=\begin{bmatrix}
	\tilde{\tA} & \mathtt{0} \\
	\mathtt{0}           & \mathtt{0}
	\end{bmatrix}$ and
	$\tB'=\begin{bmatrix}
	\tilde{\tB} & \mathtt{0} \\
	\mathtt{0}           & \mathtt{0}
	\end{bmatrix}$, and $P' = \diag(\tilde{P},I_{l-l'}), Q' = \diag(\tilde{Q},I_{m-m'})$ and $R' = \diag(\tilde{R},I_{n-n'})$. Then we set $P$ to be $T_1^{-1}P'S_1$, $Q$ to be $S_2Q'T_2^{-1}$ and $R$ to be $T_3R'S_3^{-1}$, where $P,Q$ and $R$ are unitary matrices. It's easy to check that $P\tA Q=\tB^{R}$.
	
	For the only if direction, suppose $P\tA Q=\tB^R$ for $P\in\mathrm{U}(l,\C),Q\in\mathrm{U}(m,\C)$ and $R\in\mathrm{U}(n,\C)$, which follows that $P'\tA'Q'=\tB'^{R'}$ for $\tA'=\begin{bmatrix}
	\tilde{\tA} & \mathtt{0} \\
	\mathtt{0}           & \mathtt{0}
	\end{bmatrix}$ and
	$\tB'=\begin{bmatrix}
	\tilde{\tB} & \mathtt{0} \\
	\mathtt{0}           & \mathtt{0}
	\end{bmatrix}$, and $P' = T_1PS_1^{-1}, Q' = S_2^{-1}QT_2$, and $R' = T_3^{-1}RS_3$. Write $P'$ as $\begin{bmatrix}
	P_{1,1} & P_{1,2} \\
	P_{2,1} & P_{2,2}
	\end{bmatrix}$ where $P_{1,1}$ is of size $l'\times l'$. Observe that the last $l-l'$ horizontal slices of $\tA'Q'$ and $\tB'^{R'}$ are $\mathbf{0}$ and the first $l'$ slices of $\tA'Q'$ are linearly independent, so we derive that $P_{2,1} = \mathbf{0}$. We can conclude that $Q'$ and $R'$ are block-lower-trianglular matrices in the same way. Therefore, $\tilde{P},\tilde{Q}$ and $\tilde{R}$ are unitary, where $\tilde{P}$ is the first $l'\times l'$ submatrix of $P'$, $\tilde{Q}$ is the first $m'\times m'$ submatrix of $Q'$ and $\tilde{R}$ is the first $n'\times n'$ submatrix of $R'$. Thus, $\tilde{P},\tilde{Q}$ and $\tilde{R}$ form a unitary isomorphism between $\tilde{\tA}$ and $\tilde{\tB}$ by $\tilde{P} \tilde{\tA} \tilde{Q} = \tilde{\tB}^{\tilde{R}}$.
\end{proof}

\begin{corollary}
	\label{obs: tensor to mspace}
	Given two 3-tensors $\tA,\tB\in V\otimes V\otimes W$, there is a reduction $r$ such that $\tA,\tB$ are unitarily isomorphic if and only if $r(\tA),r(\tB)\in V\otimes V\otimes W'$ are unitarily pseudo-isometric bilinear forms, and such that the frontal slices of $r(\tA)$ and $r(\tB)$ are linearly independent.
\end{corollary}

Based on Lemma \ref{obs: non-deg}, we will show that the \algprobm{Unitary $3$-Tensor Isomorphism} ($U\otimes V\otimes W$) can be reduced to \algprobm{Unitary Matrix Space Conjugacy} ($V'\otimes V'^{*}\otimes W'$).\footnote{We note that there is some ambiguity in the name here, which where the notation helps. Namely, ``unitary conjugacy of matrix spaces'' could mean either the action of $\U(V') \times \U(W')$ on $V' \otimes V'^{*} \otimes W'$ or the action of $\U(V') \times \GL(W')$ on the same space. In this paper we do not consider such ``mixed'' actions, though they are certainly interesting for future research. As a mnemonic, if we think of the matrix space itself as ``unitary'', in the sense of having a unitary structure, this lends itself to the interpretation of $\U(V') \times \U(W')$ acting.}

\begin{proposition}
	\label{lem: act1 to act3}
	There is a reduction $r: U\otimes V\otimes W\to V'\otimes V'^{*}\otimes W$ where $\dim(U) = l, \dim(V) = m, \dim(W) = n$ and $\dim(V') = l + m$ such that two tensors $\tA,\tB \in U\otimes V\otimes W$ are unitarily isomorphic if and only if $r(\tA),r(\tB) \in V'\otimes V'^{*}\otimes W$ are unitarily conjugate matrix spaces.
\end{proposition}
\begin{proof}
	\paragraph{The reduction.} Denote the $i$th frontal slice of $\tA$ by $A_i$. We construct the reduction in the following way:
	\begin{align*}
	\hat{A}_i = \begin{bmatrix}
	\mathbf{0} & A_i        \\
	\mathbf{0} & \mathbf{0}
	\end{bmatrix},
	\end{align*}
	where $\hat{A}_i \in \mathrm{M}(l+m,\C)$ is the $i$th frontal slice of $r(\tA)$.
	
	Without loss of generality, we can always assume $\tA$ and $\tB$ are non-degenerate. Then we will show that $\tA$ and $\tB$ are isomorphic if and only if $r(\tA)$ and $r(\tB)$ are isomorphic.
	
	\paragraph{For the if direction.} We assume that $r(\tA)$ and $r(\tB)$ are unitarily isomorphic, so there are $P\in\mathrm{U}(l+m,\C)$ and $Q\in\mathrm{U}(n,\C)$ such that $P^{-1}r(\tA)P=r(\tB)^Q$. Let $P$ be a block matrix:
	\begin{align*}
	\begin{bmatrix}
	P_{1,1} & P_{1,2} \\
	P_{2,1} & P_{2,2}
	\end{bmatrix},
	\end{align*}
	where $P_{1,1}$ is of size $l\times l$. Let $r(\tB)^Q$ be $r(\tB)'$ and the $i$th frontal slice of $r(\tB)'$ be $B_i'$. Since $r(\tA)P=Pr(\tB)'$, we have that
	\begin{align*}
	\begin{bmatrix}
	A_iP_{2,1} & A_iP_{2,2} \\
	\mathbf{0} & \mathbf{0}
	\end{bmatrix}=
	\begin{bmatrix}
	\mathbf{0} & P_{1,1}B_i' \\
	\mathbf{0} & P_{2,1}B_i'
	\end{bmatrix},
	\end{align*}
	where $A_iP_{2,1}=\mathbf{0}$ and $A_iP_{2,2}=P_{1,1}B_i'$ for all $i\in\left[ n \right]$. It follows that every row of $P_{2,1}$ is in the intersection of right kernels of $A_i$. Since $\tA$ is non-degenerate, $P_{2,1}$ must be a zero matrix. Thus, $P$ is a block-upper-trianglular matrix, which results in $P_{1,1}$ and $P_{2,2}$ are unitary. Therefore, we have that $P_{1,1}^{-1}\tA P_{2,2}=\tB^Q$ for $P_{1,1}\in\mathrm{U}(l,\C),P_{2,2}\in\mathrm{U}(m,\C)$ and $Q\in\mathrm{U}(n,\C)$.
	
	\paragraph{For the only if direction.} Suppose $P\tA Q=\tB^R$ where $P\in\mathrm{U}(l,\C), Q\in\mathrm{U}(m,\C)$ and $R\in\mathrm{U}(n,\C)$. Then we define $P'$ and $Q'$ as follows
	\begin{align*}
	P'=\begin{bmatrix}
	P^{-1}     & \mathbf{0} \\
	\mathbf{0} & Q
	\end{bmatrix}\quad\text{and}\quad Q'=R,
	\end{align*}
	where $P'$ and $R'$ are unitary. We can straightforwardly check that ${P'}^{-1}r(\tA)P'=r(\tB)^{Q'}$.
\end{proof}

We can similarly apply the strategy in this section to construct the reduction from \algprobm{Unitary $3$-Tensor Isomorphism} ($U\otimes V\otimes W$) to \algprobm{Bilinear Form Unitary Pseudo-isometry} ($V\otimes V\otimes W$). We record this as the following result.

\begin{proposition}
	There is a reduction $r: U\otimes V\otimes W\to V'\otimes V'\otimes W$ where $\dim(U) = l, \dim(V) = m, \dim(W) = n$ and $\dim(V') = l + m$ such that two tensors $\tA,\tB \in U\otimes V\otimes W$ are unitarily isomorphic if and only if $r(\tA),r(\tB) \in V'\otimes V'\otimes W$ are unitarily pseudo-isometric bilinear forms.
\end{proposition}

\subsection{Reduction from \algprobm{Unitary $3$-Tensor Isomorphism} to \algprobm{Unitary Algebra Iso.} ($V\otimes V\otimes V^{*}$) and \algprobm{Unitary Equivalence of Noncommutative Cubic Forms}  ($V\otimes V\otimes V$)}\label{subsec:UUWtoUUU}

\begin{proposition}
	There is a reduction from \algprobm{Bilinear Form Unitary Pseudo-isometry} to \algprobm{Unitary Algebra Isomorphism} and to \algprobm{Unitary Equivalence of Noncommutative Cubic Forms}.
	
	In symbols, there are reductions
	$$
	r \colon V\otimes V\otimes W\to V'\otimes V'\otimes V'^{*} \quad \text{and} \quad r' \colon V\otimes V\otimes W\to V'\otimes V'\otimes V'
	$$
	where $\dim(V')=\dim(V)+\dim(W)$ such that two bilinear forms $\tA,\tB \in V \otimes V \otimes W$ are unitarily pseudo-isometric if and only if $r(\tA)$ and $r(\tB)$ are unitarily isomorphic algebras, if and only if $r'(\tA)$ and $r'(\tB)$ are unitarly equivalent noncommutative cubic forms.
\end{proposition}

\begin{proof}
	\paragraph{The construction.} Given a tensor $\tA\in V\otimes V\otimes W$ whose frontal slices are $A_i$, construct an array $\tA'\in\T((l+m)\times (l+m)\times (l+m),\C)$ of which the frontal slices are
	\begin{align*}
	A'_i = \mathbf{0}~\text{for}~i\in\left[ l \right]\quad \text{and} \quad A'_i = \begin{bmatrix}
	A_{i-l}    & \mathbf{0} \\
	\mathbf{0} & \mathbf{0}
	\end{bmatrix}~\text{for}~i\in\left[ l+1,l+m \right].
	\end{align*}
	Let $\hat{\tA}$ represent the tensor in $V'\otimes V'\otimes V'^{*}$ corresponding to entries defined by $\tA'$, and denote $\tilde{\tA}$ by the tensor in $V'\otimes V'\otimes V'$ corresponding to entries defined by $\tA'$. Note that by Corollary \ref{obs: tensor to mspace}, we can always assume that the frontal slices of $\tA$ are linearly independent, so the last $m$ slices of $\tA'$ are linearly independent as well. We will show that $\tA,\tB\in V\otimes V\otimes W$ are isomorphic if and only if $\hat{\tA},\hat{\tB}\in V'\otimes V'\otimes V'^*$ are isomorphic, and $\tA,\tB$ are isomorphic if and only if $\tilde{\tA},\tilde{\tB}\in V'\otimes V'\otimes V'$ are isomorphic.
	
	\paragraph{The only if direction.} Given $P\in\mathrm{U}(l,\C)$ and $Q\in\mathrm{U}(m,\C)$ such that $P^t\tA P=\tB^Q$, set $\hat{P}$ and $\tilde{P}$ to be $\diag(P,Q^t)$ and $\diag(P,Q^{-1})$ respectively, where $\hat{P}$ and $\tilde{P}$ are unitary. Then we can straightforwardly derive that $\hat{P}^t\hat{\tA}\hat{P}=\hat{\tB}^{\hat{P}^t}$ and $(\tilde{P}^t\tilde{\tA}\tilde{P})^{\tilde{P}}=\tilde{\tB}$.
	
	\paragraph{The if direction.} We first consider the $V'\otimes V'\otimes V'^{*}$ case. Assume there is a matrix $P\in\mathrm{U}(l+m,\C)$ such that $P^t\hat{\tA}P=\hat{\tB}^{P^{t}}$. Then we write $P$ as $\begin{bmatrix}
	P_{1,1} & P_{1,2} \\
	P_{2,1} & P_{2,2}
	\end{bmatrix}$, where $P_{1,1} \in \mathrm{M}(l,\C)$. Consider the first $l$ slices $B''_i$ of $\hat{\tB}^{P^{t}}$,
	\begin{align*}
	B''_i=P^t\hat{\tA}_iP=\mathbf{0}.
	\end{align*}
	Since the last $m$ slices of $\hat{\tA}$ are linearly independent, we will have that $P_{2,1} = \mathbf{0}$. It follows that $P_{1,1}$ and $P_{2,2}$ are unitary. The equivalence of the last $m$ slices of $P^t\hat{\tA}P$ and $\hat{\tB}^{P^{t}}$ yields that $P_{1,1}^t\tA P_{1,1}=\tB^{P_{2,2}^t}$, which completes the proof of the if direction for $V'\otimes V'\otimes V'^*$.
	
	The proof for the if direction of $V'\otimes V'\otimes V'$ case is similar to the above.
\end{proof}

\section{Proof of Theorem~\ref{thm:dto3}}\label{sec:dto3}

We present the proof for unitary groups, and the argument is essentially the same for orthogonal groups.

Let $\tA, \tB$ be two $d$-way arrays in $\T(n_1\times\dots\times n_d, \F)$. We will exhibit an algorithm $T$ such that $T(\tA)$ is an algebra on $\F^m$ where $m=\poly(n_1, \dots, n_d)$, and such that $\tA$ and $\tB$ are unitarily isomorphic as $d$-tensors if and only if $T(\tA)$ and $T(\tB)$ are unitarily isomorphic as algebras. We can then apply Theorem~\ref{thm:5actions} to reduce to \algprobm{Unitary $3$-Tensor Isomorphism}. Therefore, in the following we focus on the step of reducing \algprobm{Unitary $d$-Tensor Isomorphism} to \algprobm{Unitary Algebra Isomorphism}.

\paragraph{Background on quivers and path algebras.} A \emph{quiver} is a directed multigraph $G=(V, E, s, t)$, where $V$ is the vertex set, $E$ is the arrow
set, and
$s, t:E\to V$ are two maps indicating the source and target of an arrow. 

A path in $G$ is 
the concatenation of edges $p=e_1,e_2,\ldots, e_n$, where $e_i\in E$ for $i\in
[n]$, such that $s(e_{i+1})=t(e_i)$ for $i\in[n-1]$. $s(p)=s(e_1)$ is the source
of $p$, $t(p)=t(e_n)$ is the target of $p$ and $l(p)=n$ is
the length of $p$. For
a consistent notation including the vertex, we define the source $s(v)$ and target
$t(v)$ for each vertex $v\in V$ by $s(v)=t(v)=v$, and we regard the length $l(v)$
of every vertex $v$ as $0$.
Note that $V$ consists of
paths of length $0$, and $E$ consists of paths of length $1$.

Let $\F$ be a field. The \emph{path algebra} of $G$, denoted as $\Path_\F(G)$, is
the free algebra generated by $V\cup E$ modulo the relations generated by:
\begin{enumerate}
	\item For $v, v'\in V$, $vv'=v$ if $v=v'$, and $0$ otherwise.
	\item For $v\in V$ and $e\in E$, $ve=e$ if $v=s(e)$, and $0$ otherwise. And $ev=e$
	if $v=t(e)$, and $0$ otherwise.
	\item For $e, e'\in E$, $ee'=0$ if $t(e)\neq s(e')$.
\end{enumerate}

In this paper we make use of the following quiver. Note that this is different from the quiver used in \cite{GQ21}; this difference leads to some significant simplifications in the argument, and allows the argument to go through for unitary and orthogonal groups (it is unclear to us whether the original argument in \cite{GQ21} does so).
\begin{figure}[!htbp]
	\[
	\xymatrix{
		v_1 \ar[r] \ar@/^2pc/[r]^{x_{1,1}} \ar@/^1pc/[r]^{x_{1,2}} \ar[r]^{...}
		\ar@/_/[r]_{x_{1,n_1}}
		&
		v_2 \ar[r] \ar@/^2pc/[r]^{x_{2,1}} \ar@/^1pc/[r]^{x_{2,2}} \ar[r]^{...}
		\ar@/_/[r]_{x_{2,n_2}} &
		v_3 \ar[r] \ar@/^2pc/[r]^{x_{3,1}} \ar@/^1pc/[r]^{x_{3,2}} \ar[r]^{...}
		\ar@/_/[r]_{x_{3,n_3}} &
		\dotsb \ar[r] \ar@/^2pc/[r]^{x_{d-1,1}} \ar@/^1pc/[r]^{x_{d-1,2}} \ar[r]^{...}
		\ar@/_/[r]_{x_{d-1,n_{d-1}}} &
		v_d \ar[r] \ar@/^2pc/[r]^{x_{d,1}} \ar@/^1pc/[r]^{x_{d,2}} \ar[r]^{...}
		\ar@/_/[r]_{x_{d,n_{d-1}}} &
		v_{d+1}
		\\
	}
	\]
	\caption{ \label{fig:graph} The quiver $G$ we use in this paper.}
\end{figure}
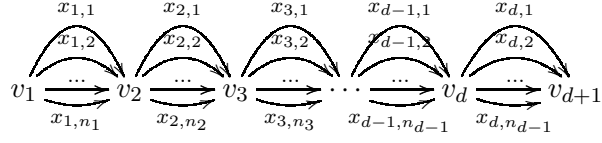
Note that $G=(V, E, s, t)$ where $V=\{v_1, \dots, v_{d+1}\}$, $E=\{x_{i, j}\mid
i\in[d], j\in[n_i]\}$, $s(x_{i, j})=v_i$ and $t(x_{i, j})=v_{i+1}$.

\begin{proof}[Proof of Theorem~\ref{thm:dto3}]
	Let $f, g\in U_1\otimes U_2\otimes\dots\otimes U_d$ be two tensors, where
	$U_i=\F^{n_i}$ for $i\in[d]$. We can encode $f$ in
	$\Path_\F(G)$ as follows. Recall that $e_i$ denotes the $i$th standard basis
	vector. Suppose $f=\sum_{(i_1, \dots, i_d)}\alpha_{i_1, \dots, i_d}
	e_{i_1}\otimes \dots \otimes e_{i_d}$, where the summation is over $(i_1, \dots,
	i_d)\in [n_1]\times
	\dots \times [n_d]$ and $\alpha_{i_1, \dots, i_d}\in \F$. Then let $\hat{f}\in
	\Path_\F(G)$ be
	defined as $\hat{f}=\sum_{(i_1, \dots, i_d)}\alpha_{i_1, \dots, i_d}x_{1, i_1}
	x_{2, i_d} \dots x_{d, i_d}$, where $(i_1, \dots, i_d)\in [n_1]\times \dots \times
	[n_d]$.
	
	Let $R_f:=\Path_\F(G)/(\hat{f})$ and $R_g:=\Path_\F(G)/(\hat{g})$. We will show
	that $f$
	and $g$ are unitarily isomorphic as tensors if and only if $R_f$ and $R_g$ are unitarily isomorphic as
	algebras.
	
	\paragraph{Tensor isomorphism implies algebra isomorphism.} Let $(P_1, \dots,
	P_d)\in \mathrm{U}(n_1, \C)\times \dots\times\mathrm{U}(n_d, \C)$ be a tensor isomorphism from
	$f$ to $g$. Then
	$P_i$ naturally acts on the linear space $\langle x_{i, 1}, \dots, x_{i,
		n_i}\rangle$. Together with the identity matrix $I_{d+1}$ acting on $\langle v_1,
	\dots, v_{d+1}\rangle$, we claim that they form an algebra isomorphism from $R_f$
	to $R_g$.
	
	We will show that this is a homomorphism, and then verify that it is indeed an isomorphism. This part is essentially the same as \cite{GQ21}. 
	
	To show that it is a homomorphism, we first examine the quiver relations. This homomorphism $R_f\to R_g$ is induced by a linear map, this map $P$ is defined by $P(v_i)=v_i$ and
	\begin{align*}
	P(x_{i,j})=\sum_{k=1}^{n_i} (P_i)_{jk}y_{i,k}\quad\text{for}~i=1,\ldots,d-1,
	\end{align*}
	where $y_{1,1},\ldots,y_{d,n_d},v_1,\ldots,v_d$ denote generators of $R_g$. Let $x_{1,1},\ldots,x_{d,n_d},v_1,\ldots,v_d$ be generators of $R_f$, and then the following quiver relations need to be checked:
	\begin{align*}
	v_iv_{i'}        & = \delta_{i,i'}v_i            \\
	v_ix_{i',j}      & = \delta_{i,i'}x_{i',j}       \\
	x_{i,j}v_{i'}    & = \delta_{i+1,i'}x_{i,j}      \\
	x_{i,j}x_{i',j'} & = 0\quad\text{if}~i+1\neq i'.
	\end{align*}
	
	It's not hard to examine the first three which involve the $v_i$, as 
	\begin{align*}
	P(v_iv_{i'})     & = P(v_i)P(v_{i'})  = v_iv_{i'}   = \delta_{i,i'}P(v_i),                        \\
	P(v_ix_{i',j})   & = P(v_i)P(x_{i',j})   = v_i \sum_{k=1}^{n_{i'}} (P_i')_{jk}y_{i',k}  = \delta_{i,i'}P(x_{i',j}),                   \\
	P(x_{i,j}v_{i'}) & = P(x_{i,j})P(v_{i'})   = \sum_{k=1}^{n_i} (P_i)_{jk}y_{i,k}v_{i'}   = \delta_{i+1,i'}P(x_{i,j}).
	\end{align*}
	For the last relation,
	\begin{align*}
	P(x_{i,j}x_{i',j'}) & = P(x_{i,j})P(x_{i',j'})                                                       \\
	& = \sum_{k=1}^{n_i} (P_i)_{jk}y_{i,k}\sum_{k=1}^{n_{i'}} (P_{i'})_{j'k}y_{i',k} \\
	& = \sum_{k=1}^{n_i}\sum_{k=1}^{n_{i'}}(P_i)_{jk}(P_{i'})_{j'k}y_{i,k}y_{i',k}   \\
	& = 0 \quad \text{if}~i+1\neq i'.
	\end{align*}
	Therefore, the map $R_f\to R_g$ induced by $P$ is an algebra homomorphism.
	
	Let $n=\max_i\{n_i\}$. To prove the map $R_f\to R_g$ is an algebra isomorphism, it requires to check the dimension of $R_f$ first:
	\begin{align*}
	\dim(R_f) & = \#\{v_i\} + \sum_{i=1}^{d}\sum_{j=0}^{d-i}\#\{\text{paths from}~v_i~\text{to}~v_{i+j}\} \\
	& = d + \sum_{i=1}^{d}\sum_{j=0}^{d-i}\prod_{k=i}^{i+j}n_j                                  \\
	& \leq d + \sum_{i=1}^{d}\sum_{j=0}^{d-i}n^d                                                \\
	& \leq O(d^2n^d).
	\end{align*}
	If $d$ is fixed, the dimension of $R_f$ is polynomial with $n$. Next, as $P$ is an isomorphism of ($\sum_{i=1}^{d}n_i+d$)-vector spaces, it follows that $R_f\to R_g$ induced by $P$ is surjective on all the generators of $R_g$ and hence it's surjective on the whole $R_g$. Finally, since $\dim(R_f)=\dim(R_g)<\infty$, the map $R_f\to R_g$ is a bijection and it's naturally an algebra isomorphism from $R_f$ to $R_g$.
	
	\paragraph{Algebra isomorphism implies tensor isomorphism.} This part of the proof is new, compared to the corresponding part in \cite{GQ21}.
	
	Let
	$\phi:\Path_\F(G)/(\hat{f})\to \Path_\F(G)/(\hat{g})$ be an algebra isomorphism,
	which is determined by the images of $v_i$, $x_{j, k}$ under $\phi$.
	
	Note that $\Path_\F(G)$ is linearly spanned by paths in $G$, so it is naturally
	graded, and
	we use $\Path_\F(G)_\ell$ denotes the linear space of $\Path_\F(G)$ spanned by
	paths of length exactly $\ell$.
	
	First, note that $\phi(\hat{f})=\alpha\cdot \hat{g}+\text{a linear combination of
		quiver relations}$, where $\alpha\in \F$.
	
	Second, we claim that the coefficient of $v_i$ in $\phi(x_{j, k})$ must be zero
	for
	any $i, j, k$. If not, suppose $\phi(x_{j, k})=\gamma\cdot v_i+M$ where
	$\gamma\neq 0$, and $M$
	denotes other terms not containing $v_i$. On the one hand, $\phi(x_{j, k}^2)=0$
	because $x_{j, k}^2=0$ by the quiver relations. On the other hand, $\phi(x_{j,
		k})^2=(\gamma\cdot
	v_i+M)^2=\gamma^2\cdot v_i^2+M'=\gamma^2\cdot v_i+M'$ where $M'$ denotes other
	terms, which cannot
	contain $v_i$. So $\phi(x_{j, k})^2$ is nonzero, contradicting $\phi(x_{j,
		k}^2)=0$ and $\phi$ being an algebra isomorphism.
	
	By the above, it follows for any path $P$ (a product of $x_{i, j}$'s) of length
	$\ell\geq 1$, $\phi(P)$ is a linear combination of paths of length $\geq \ell$.
	This implies that, if we express $\phi$ in the linear basis of
	$\Path_\F(G)/(\hat{f})$, $(v_1, \dots, v_{d+1}, x_{i, j}, \text{paths of length
	}2, \dots, \text{paths of length }d)$, then $\phi$ is a block-lower-triangular
	matrix, where the each block is determined by the path lengths. That is, the first block is indexed by $(v_1, \dots, v_{d+1})$, the second block is indexed by $(x_{i, j})$, the third block is indexed by paths of length $2$, and so on. 
	
	Third, we claim that for $1\leq i<j\leq d+1$, the coefficient of $x_{i, k}$ in
	$\phi(x_{j, k'})$ must be zero. If not, then let $P$ be a path of length $d-i$
	starting from $v_{i+1}$. Because of the block-lower-triangular matrix structure and that $\phi$ is an isomorphism,
	we know that there exists a path $P'$ of
	length $d-i$, such that the coefficient of $P$ in $\phi(P')$ is nonzero. Then
	$\phi(x_{j, k'}\cdot P')=\phi(x_{j, k'})\cdot \phi(P')=(\beta\cdot x_{i,
		k}+M)\cdot
	(\gamma\cdot P+N)=\beta\cdot \gamma\cdot x_{i, k}\cdot P+ L$, where $M$, $N$ and
	$L$ denote appropriate
	other terms, and $\beta, \gamma\in \F$ are non-zero. Note that $x_{i, k}\cdot P$
	cannot be cancelled from other terms. This implies that $\phi(x_{j, k'}\cdot P')$
	is non-zero. However, $x_{j, k'}\cdot P'$ has to be zero because $P'$ is of length
	$d-i$, so it starts from some variable $x_{i+1, k''}$. This leads to the desired
	contradiction.
	
	By the above, if we restrict $\phi$ to the linear subspace $\langle x_{i,
		j}\rangle$ in the linear basis
	$$(x_{1, 1}, \dots, x_{1, n_1}, \dots, x_{d, 1},
	\dots, x_{n_d}), $$
	then $\phi$ is again in the block-lower-triangular form,
	where the blocks are determined by the first index of $x_{i, j}$. That is, the first block is indexed by $x_{1, j}$ for all $j$, the second block is indexed by $x_{2, j}$ for all $j$, and so on.
	
	We now can take the diagonal block of $\phi$ on $(x_{i, 1}, \dots, x_{i, n_i})$,
	and let the resulting (invertible) matrix be $P_i$. These matrices $P_1, \dots,
	P_d$ together determine a linear map $\psi$ on $\langle x_{i, j}\rangle$. By
	comparing degrees, we see that $\psi(\hat{f})=\alpha\cdot \hat{g}$. Now suppose
	$\F$ contains $d$th roots. We can then obtain $(1/\alpha^{1/d}\cdot P_1,
	1/\alpha^{1/d}\cdot P_2, \dots, 1/\alpha^{1/d}\cdot P_d)\cdot f=g$.
	
	Getting back to our original goal, we see that if $\psi$ is unitary, then the
	block-lower-triangular form of $\psi$ implies that it is actually block-diagonal,
	and the diagonal blocks are all unitary as well. This shows that $P_i$'s are
	unitary, and $f$ and $g$ are unitarily isomorphic.
\end{proof}

\appendix

\section{Polynomial systems for \algprobm{Tensor Isomorphism} and related problems}\label{app:groebner}

We provide more details for the Gr\"obner basis experiments described in Section~\ref{subsec:motivation}. 

Let us first examine how to formulate \algprobm{Tensor Isomorphism} as solving a system of polynomial equations. Let $\tA, \tB\in \T(n\times n\times n, \F)$ be two 3-way arrays. Let $X, Y, Z$ be three $n\times n$ variable matrices. Then $X\tA^Z Y=\tB$ can be seen as encoding $n^3$ many cubic polynomials in $3n^2$ many variables, whose coefficients are determined by the entries of $\tA$ and $\tB$. To ensure that $X, Y, Z$ are invertible, we introduce new variables $x, y, z$, and include polynomials $\det(X)\cdot x=1$, $\det(Y)\cdot y=1$, and $\det(Z)\cdot z=1$.  (This is similar to \cite{GGPS}, although there instead of using the determinant they introduce twice as many new variables, and equations $XX'=X'X=I$ and similarly for $Y$ and $Z$.  This reduces degree compared to our equations here, but at the expense of many more variables). 
This gives a system of polynomial equations, which has a solution over $\F$ if and only if $\tA$ and $\tB$ are isomorphic as tensors over $\F$. Then this problem can be solved by e.g. the Gr\"obner basis algorithm.

Of course the above is just one approach. Indeed, from the Gr\"obner basis viewpoint, it is more desirable to consider $X\tA Y=\tB^Z$ so we get quadratic equations rather than cubic ones. Interested readers may refer \cite{TangDJPQS22} for more optimisations as such.

The instances we do experiments on are drawn as follows. Note that if $\tA$ and $\tB$ are both random instances, then with high probability they are not isomorphic. To get isomorphic pairs instead, we can sample a random $\tA$ and random invertible matrices $R, S, T$, and set $\tB=(R, S, T)\circ \tA$. (This is the setting used in cryptographic schemes based on $\cc{TI}$-hardness, e.\,g., \cite{JQSY19}.) The pair $(\tA, \tB)$ is then set as the input.

If orthogonal isomorphism is needed, we can set $X^t X=I$ which is a system of quadratic equations. We can also sample a random orthogonal matrix over $\F_q$ by existing functionality of Magma. 

We now introduce the exact problem to be tackled by our actual experiments. Let $\phi, \psi:\F_q^n\times\F_q^n\times\F_q^n\to\F_q$ be two alternating trilinear forms. We say that $\phi, \psi$ are isomorphic, if there exists $A\in\GL(n, q)$, such that $\phi(Au, Av, Aw)=\psi(u,v,w)$ for any $u, v, w\in \F_q^n$. To decide whether two alternating trilinear forms are isomorphic is known to be $\cc{TI}$-complete \cite{GQT22}. This problem can be similarly formulated as solving systems of polynomial equations, and some technical issues also follow the ideas as described above.

\section{The proof outline of Theorem~\ref{thm:FGS_parameter}}\label{app:proof}

In this subsection we give an outline for the proof of Theorem~\ref{thm:FGS_parameter}. The goal here is to give a guided exposition of some main technical steps in the proof of \cite[Theorem 1.1]{FGS19}, so the reader may verify the parameters in conjunction with \cite{FGS19} more easily. This requires us to examine the constructions in \cite{FGS19} to
compute the parameters explicitly. 

\subsection{Step 1: Block isomorphism and plain
	isomorphism}\label{subsubsec:block-iso}

The first notion is the block isomorphism of $3$-tensors. Let $\tA$ and $\tB$
be two plain $3$-tensors in $U\otimes V\otimes W$. Let $U=U_1\oplus \dots \oplus
U_e$, $V=V_1\oplus \dots \oplus V_f$, and $W=W_1\oplus \dots \oplus W_g$ be direct
sum decompositions. Let $\cE\leq\GL(U)$ be the subgroup of $\GL(U)$ that preserves
this direct sum decomposition, that is, $E$ consists of those invertible linear
maps that sends $U_i$ to $U_i$ for every $i\in[e]$. Similarly let $\cF$ (resp. $\cG$)
be the subgroup of
$\GL(V)$ (resp. $\GL(W)$) preserving the direct sum decomposition. We say that $\tA$
and $\tB$ are \emph{block-isomorphic} with respect to these direct sum
decompositions if $\tA$ and $\tB$ are in the same orbit under $\cE\times \cF\times \cG$.

\begin{proposition}[{Rephrase of \cite[Theorem 2.1]{FGS19}}]\label{prop:block}
	Let $U=U_1\oplus \dots \oplus
	U_e$, $V=V_1\oplus \dots \oplus V_f$, and $W=W_1\oplus \dots \oplus W_g$ be direct
	sum decompositions of vector spaces $U$, $V$, and $W$. Then there exists an
	algorithm $B$ that takes $\tA, \tB\in U\otimes V\otimes W$ and outputs vector spaces
	$U', V', W'$ and $B(\tA), B(\tB)\in U'\otimes V'\otimes W'$ such that $\tA$ and $\tB$ are
	block-isomorphic if and only if $B(\tA)$ and $B(\tB)$ are isomorphic. The algorithm
	runs in time polynomial in the maximum dimension over $U, V, W$, and this maximum
	dimension is upper bounded
	by $\poly(\dim(U), \dim(V), \dim(W), 2^e, 2^f, 2^g)$.
\end{proposition}

\subsection{Step 2: Linked-block isomorphism and block isomorphism}

The second notion is the linked-block isomorphism of $3$-tensors. Again, let $\tA$
and $\tB$
be two plain $3$-tensors in $U\otimes V\otimes W$. Let $U=U_1\oplus \dots \oplus
U_e$, $V=V_1\oplus \dots \oplus V_f$, and $W=W_1\oplus \dots \oplus W_g$ be direct
sum decompositions.
Let $\cE\leq\GL(U)$, $\cF\leq\GL(V)$ and $\cG\leq\GL(W)$ be defined as in
Section~\ref{subsubsec:block-iso}.

Let $I_U=[e]$, $I_V=[f]$, and $I_W=[g]$. Suppose two binary
relations $\sim$ and $\Join$ on $I_U\cup I_V\cup I_W$ satisfy the
following: (1) $\sim$ is an equivalence relation; (2) if $a\Join b$ then $a\not\sim
b$; and (3) if $a\Join b$, then $b\Join c\iff a\sim c$.

For convenience, we shall use $X_a$ to denote $U_a$, $V_a$, or $W_a$ depending on
whether $a\in I_U$, $a\in I_V$, or $a\in I_W$. Briefly speaking, $a \sim
b$ denotes that the corresponding two blocks are acted covariantly, and $a\Join b$
denotes that the corresponding two blocks are acted contravariantly. So if $a\sim
b$ or $a\Join b$, then $\dim(X_a)=\dim(X_b)$.

Given such binary relations $\sim$ and $\Join$, we define a block-isomorphism $X$ between $\tA$ and $\tB$ to be a \textit{linked-block-isomorphism} if for any $a,b\in I_U\cup I_V\cup I_W$, the following conditions for decompositions of $U,V$ and $W$ holds:
\begin{align*}
X_a=X_b~\text{if}~a\sim b,\qquad X_a=X_b^{-t}~\text{if}~a\Join b.
\end{align*}
\begin{proposition}[{Rephrase of \cite[Theorem 4.1]{FGS19}}]\label{prop: link-block}
	Let $U=U_1\oplus \dots \oplus
	U_e$, $V=V_1\oplus \dots \oplus V_f$, and $W=W_1\oplus \dots \oplus W_g$ be direct
	sum decompositions of vector spaces $U$, $V$, and $W$ and these decompositions satisfy conditions with respect to some binary relations $\sim$ and $\Join$ for $I_U\cup I_V\cup I_W$. Then there exists an
	algorithm $B$ that takes $\tA, \tB\in U\otimes V\otimes W$ and outputs vector spaces
	$U', V', W'$ and $B(\tA), B(\tB)\in U'\otimes V'\otimes W'$, where $U'=U_1'\oplus \dots \oplus U_{\poly(e)}'$, $V'=V_1'\oplus \dots \oplus V_{\poly(f)}'$, and $W'=W_1'\oplus \dots \oplus W_{\poly(g)}'$ such that $\tA$ and $\tB$ are
	linked-block-isomorphic if and only if $B(\tA)$ and $B(\tB)$ are block-isomorphic. The algorithm
	runs in time polynomial in the maximum dimension over $U, V, W$, and this maximum
	dimension is upper bounded
	by $\poly(e,f,g)\cdot\max(\dim(U),\dim(V),\dim(W))$.
\end{proposition}

\paragraph{Acknowledgement.} 
J.A.G. was supported by NSF CAREER grant CISE-2047756. Y.Q. was partially supported by the Australian Research Council DP200100950 and LP220100332. G. T. was partially supported by the Australian Research Council LP220100332 and the Sydney Quantum Academy, Sydney, NSW, Australia. C. Z. was partially supported by the Australian Research Council DP200100950 and the Sydney Quantum Academy, Sydney, NSW, Australia. Z. C. was partially supported by the National Research Foundation, Singapore under its CQT (Centre for Quantum Technologies) Bridging Grant.

\bibliographystyle{alphaurl}
\bibliography{references}

\newcommand{\etalchar}[1]{$^{#1}$}
\begin{thebibliography}{DLDMV00}

\bibitem[ABLS01]{ABLS01}
Antonio Ac{\'\i}n, Dagmar Bru{\ss}, Maciej Lewenstein, and Anna Sanpera.
\newblock Classification of mixed three-qubit states.
\newblock {\em Physical Review Letters}, 87(4):040401, 2001.
\newblock \href {https://doi.org/10.1103/PhysRevLett.87.040401}
  {\path{doi:10.1103/PhysRevLett.87.040401}}.

\bibitem[BJP97]{Magma}
W.~Bosma, J.~{J. Cannon}, and C.~Playoust.
\newblock The {M}agma algebra system {I}: the user language.
\newblock {\em J. Symb. Comput.}, pages 235--265, 1997.

\bibitem[BOST19]{BOS19}
Magali Bardet, Ayoub Otmani, and Mohamed Saeed-Taha.
\newblock Permutation code equivalence is not harder than graph isomorphism
  when hulls are trivial.
\newblock In {\em 2019 IEEE International Symposium on Information Theory
  (ISIT)}, pages 2464--2468. IEEE, 2019.
\newblock \href {https://doi.org/10.1109/ISIT.2019.8849855}
  {\path{doi:10.1109/ISIT.2019.8849855}}.

\bibitem[BPR{\etalchar{+}}00]{LOCC}
Charles~H Bennett, Sandu Popescu, Daniel Rohrlich, John~A Smolin, and Ashish~V
  Thapliyal.
\newblock Exact and asymptotic measures of multipartite pure-state
  entanglement.
\newblock {\em Physical Review A}, 63(1):012307, 2000.
\newblock \href {https://doi.org/10.1103/PhysRevA.63.012307}
  {\path{doi:10.1103/PhysRevA.63.012307}}.

\bibitem[CGQ{\etalchar{+}}24]{ITCS24}
Zhili Chen, Joshua~A. Grochow, Youming Qiao, Gang Tang, and Chuanqi Zhang.
\newblock On the complexity of isomorphism problems for tensors, groups, and
  polynomials {III:} actions by classical groups.
\newblock In Venkatesan Guruswami, editor, {\em 15th Innovations in Theoretical
  Computer Science Conference, {ITCS} 2024, January 30 to February 2, 2024,
  Berkeley, CA, {USA}}, volume 287 of {\em LIPIcs}, pages 31:1--31:23. Schloss
  Dagstuhl - Leibniz-Zentrum f{\"{u}}r Informatik, 2024.
\newblock URL: \url{https://doi.org/10.4230/LIPIcs.ITCS.2024.31}, \href
  {https://doi.org/10.4230/LIPICS.ITCS.2024.31}
  {\path{doi:10.4230/LIPICS.ITCS.2024.31}}.

\bibitem[CLM{\etalchar{+}}14]{LOCC_survey}
Eric Chitambar, Debbie Leung, Laura Man{\v{c}}inska, Maris Ozols, and Andreas
  Winter.
\newblock Everything you always wanted to know about {LOCC} (but were afraid to
  ask).
\newblock {\em Communications in Mathematical Physics}, 328:303--326, 2014.
\newblock \href {https://doi.org/10.1007/s00220-014-1953-9}
  {\path{doi:10.1007/s00220-014-1953-9}}.

\bibitem[D'A23]{DAlconzo}
Giuseppe D'Alconzo.
\newblock Monomial isomorphism for tensors and applications to code equivalence
  problems.
\newblock Cryptology ePrint Archive, Paper 2023/396, 2023.
\newblock URL: \url{https://eprint.iacr.org/2023/396}.

\bibitem[DLDMV00]{DDV00}
Lieven De~Lathauwer, Bart De~Moor, and Joos Vandewalle.
\newblock A multilinear singular value decomposition.
\newblock {\em SIAM journal on Matrix Analysis and Applications},
  21(4):1253--1278, 2000.
\newblock \href {https://doi.org/10.1137/S0895479896305696}
  {\path{doi:10.1137/S0895479896305696}}.

\bibitem[DSL08]{DL08}
Vin De~Silva and Lek-Heng Lim.
\newblock Tensor rank and the ill-posedness of the best low-rank approximation
  problem.
\newblock {\em SIAM Journal on Matrix Analysis and Applications},
  30(3):1084--1127, 2008.
\newblock \href {https://doi.org/10.1137/06066518X}
  {\path{doi:10.1137/06066518X}}.

\bibitem[Edm65]{Edm65}
Jack Edmonds.
\newblock Paths, trees, and flowers.
\newblock {\em Canadian Journal of mathematics}, 17(3):449--467, 1965.
\newblock \href {https://doi.org/10.4153/CJM-1965-045-4}
  {\path{doi:10.4153/CJM-1965-045-4}}.

\bibitem[EY36]{EY36}
Carl Eckart and Gale Young.
\newblock The approximation of one matrix by another of lower rank.
\newblock {\em Psychometrika}, 1(3):211--218, 1936.
\newblock \href {https://doi.org/10.1007/BF02288367}
  {\path{doi:10.1007/BF02288367}}.

\bibitem[FG11]{FortnowGrochowPEq}
Lance Fortnow and Joshua~A. Grochow.
\newblock Complexity classes of equivalence problems revisited.
\newblock {\em Inform. and Comput.}, 209(4):748--763, 2011.
\newblock Also available as arXiv:0907.4775 [cs.CC].
\newblock \href {https://doi.org/10.1016/j.ic.2011.01.006}
  {\path{doi:10.1016/j.ic.2011.01.006}}.

\bibitem[FGS19]{FGS19}
Vyacheslav Futorny, Joshua~A. Grochow, and Vladimir~V. Sergeichuk.
\newblock Wildness for tensors.
\newblock {\em Linear Algebra and its Applications}, 566:212--244, 2019.
\newblock \href {https://doi.org/10.1016/j.laa.2018.12.022}
  {\path{doi:10.1016/j.laa.2018.12.022}}.

\bibitem[GGPS23]{GGPS}
Nicola Galesi, Joshua~A. Grochow, Toniann Pitassi, and Adrian She.
\newblock On the algebraic proof complexity of {Tensor} {Isomorphism}.
\newblock In {\em Computational Complexity Conference (CCC) 2023}, 2023.
\newblock Preprint arXiv:2305.19320 [cs.CC].

\bibitem[GQ23a]{GQ21}
Joshua~A. Grochow and Youming Qiao.
\newblock On the complexity of isomorphism problems for tensors, groups, and
  polynomials {I}: {Tensor} {Isomorphism}-completeness.
\newblock {\em SIAM J. Comput.}, 52:568--617, 2023.
\newblock Part of the preprint arXiv:1907.00309 [cs.CC]. Preliminary version
  appeared at ITCS '21, DOI:10.4230/LIPIcs.ITCS.2021.31.
\newblock \href {https://doi.org/10.1137/21M1441110}
  {\path{doi:10.1137/21M1441110}}.

\bibitem[GQ23b]{TI-IV}
Joshua~A. Grochow and Youming Qiao.
\newblock On the complexity of isomorphism problems for tensors, groups, and
  polynomials {IV:} linear-length reductions and their applications.
\newblock {\em CoRR}, abs/2306.16317, 2023.
\newblock \href {http://arxiv.org/abs/2306.16317} {\path{arXiv:2306.16317}},
  \href {https://doi.org/10.48550/arXiv.2306.16317}
  {\path{doi:10.48550/arXiv.2306.16317}}.

\bibitem[GQ24]{TI2}
Joshua~A. Grochow and Youming Qiao.
\newblock On \emph{p}-group isomorphism: Search-to-decision,
  counting-to-decision, and nilpotency class reductions via tensors.
\newblock {\em {ACM} Trans. Comput. Theory}, 16(1):2:1--2:39, 2024.
\newblock Preliminary version appeared at CCC'21,
  doi.org/10.4230/LIPIcs.CCC.2021.16.
\newblock \href {https://doi.org/10.1145/3625308} {\path{doi:10.1145/3625308}}.

\bibitem[GQT22]{GQT22}
Joshua~A Grochow, Youming Qiao, and Gang Tang.
\newblock Average-case algorithms for testing isomorphism of polynomials,
  algebras, and multilinear forms.
\newblock {\em journal of Groups, Complexity, Cryptology}, 14, 2022.
\newblock Extended abstract appeared in STACS '21
  DOI:10.4230/LIPIcs.STACS.2021.38.
\newblock \href {https://doi.org/10.46298/jgcc.2022.14.1.9431}
  {\path{doi:10.46298/jgcc.2022.14.1.9431}}.

\bibitem[HQ21]{HQ20}
Xiaoyu He and Youming Qiao.
\newblock On the {Baer}--{Lov{\'{a}}sz}--{Tutte} construction of groups from
  graphs: Isomorphism types and homomorphism notions.
\newblock {\em Eur. J. Comb.}, 98:103404, 2021.
\newblock \href {https://doi.org/10.1016/j.ejc.2021.103404}
  {\path{doi:10.1016/j.ejc.2021.103404}}.

\bibitem[HU17]{HU17}
Wolfgang Hackbusch and Andr{\'e} Uschmajew.
\newblock On the interconnection between the higher-order singular values of
  real tensors.
\newblock {\em Numerische Mathematik}, 135:875--894, 2017.
\newblock \href {https://doi.org/10.1007/s00211-016-0819-9}
  {\path{doi:10.1007/s00211-016-0819-9}}.

\bibitem[Hum]{Humphreys}
Jim Humphreys.
\newblock What are ``classical groups''?
\newblock
  \url{https://mathoverflow.net/questions/50610/what-are-classical-groups}.

\bibitem[JQSY19]{JQSY19}
Zhengfeng Ji, Youming Qiao, Fang Song, and Aaram Yun.
\newblock General linear group action on tensors: {A} candidate for
  post-quantum cryptography.
\newblock In {\em Theory of Cryptography - 17th International Conference, {TCC}
  2019, Nuremberg, Germany, December 1-5, 2019, Proceedings, Part {I}}, pages
  251--281, 2019.
\newblock \href {https://doi.org/10.1007/978-3-030-36030-6_11}
  {\path{doi:10.1007/978-3-030-36030-6_11}}.

\bibitem[KST93]{KST93}
Johannes K\"{o}bler, Uwe Sch\"{o}ning, and Jacobo Tor\'{a}n.
\newblock {\em The graph isomorphism problem: its structural complexity}.
\newblock Birkhauser Verlag, Basel, Switzerland, Switzerland, 1993.
\newblock \href {https://doi.org/10.1007/978-1-4612-0333-9}
  {\path{doi:10.1007/978-1-4612-0333-9}}.

\bibitem[Lim21]{Lim21}
Lek-Heng Lim.
\newblock Tensors in computations.
\newblock {\em Acta Numerica}, 30:555--764, 2021.
\newblock \href {https://doi.org/10.1017/S0962492921000076}
  {\path{doi:10.1017/S0962492921000076}}.

\bibitem[Lov79]{Lov79}
L{\'{a}}szl{\'{o}} Lov{\'{a}}sz.
\newblock On determinants, matchings, and random algorithms.
\newblock In Lothar Budach, editor, {\em Fundamentals of Computation Theory,
  {FCT} 1979, Proceedings of the Conference on Algebraic, Arthmetic, and
  Categorial Methods in Computation Theory, Berlin/Wendisch-Rietz, Germany,
  September 17-21, 1979}, pages 565--574. Akademie-Verlag, Berlin, 1979.

\bibitem[LQW{\etalchar{+}}23]{LQWWZ22}
Yinan Li, Youming Qiao, Avi Wigderson, Yuval Wigderson, and Chuanqi Zhang.
\newblock Connections between graphs and matrix spaces.
\newblock {\em Israel Journal of Mathematics}, 256(2):513--580, 2023.

\bibitem[NC00]{NC00}
M.~Nielsen and I.~Chuang.
\newblock {\em Quantum computation and quantum information}.
\newblock Cambridge University Press, 2000.
\newblock \href {https://doi.org/10.1017/CBO9780511976667}
  {\path{doi:10.1017/CBO9780511976667}}.

\bibitem[RST22]{RST}
Krijn Reijnders, Simona Samardjiska, and Monika Trimoska.
\newblock Hardness estimates of the {Code Equivalence Problem} in the rank
  metric.
\newblock In {\em WCC 2022: The Twelfth International Workshop on Coding and
  Cryptography}, 2022.
\newblock Cryptology ePrint Archive, Paper 2022/276,
  \url{https://eprint.iacr.org/2022/276}.

\bibitem[Sei18]{Sei18}
Anna Seigal.
\newblock Gram determinants of real binary tensors.
\newblock {\em Linear Algebra and its Applications}, 544:350--369, 2018.
\newblock \href {https://doi.org/10.1016/j.laa.2018.01.019}
  {\path{doi:10.1016/j.laa.2018.01.019}}.

\bibitem[Ser98]{Ser98}
Vladimir~V Sergeichuk.
\newblock Unitary and {Euclidean} representations of a quiver.
\newblock {\em Linear Algebra and its Applications}, 278(1-3):37--62, 1998.
\newblock \href {https://doi.org/10.1016/S0024-3795(98)00006-8}
  {\path{doi:10.1016/S0024-3795(98)00006-8}}.

\bibitem[TDJ{\etalchar{+}}22]{TangDJPQS22}
Gang Tang, Dung~Hoang Duong, Antoine Joux, Thomas Plantard, Youming Qiao, and
  Willy Susilo.
\newblock Practical post-quantum signature schemes from isomorphism problems of
  trilinear forms.
\newblock In Orr Dunkelman and Stefan Dziembowski, editors, {\em Advances in
  Cryptology - {EUROCRYPT} 2022 - 41st Annual International Conference on the
  Theory and Applications of Cryptographic Techniques, Trondheim, Norway, May
  30 - June 3, 2022, Proceedings, Part {III}}, volume 13277 of {\em Lecture
  Notes in Computer Science}, pages 582--612. Springer, 2022.
\newblock \href {https://doi.org/10.1007/978-3-031-07082-2_21}
  {\path{doi:10.1007/978-3-031-07082-2_21}}.

\bibitem[Tem04]{CartesianTensor}
George Frederick~James Temple.
\newblock {\em Cartesian Tensors: an introduction}.
\newblock Courier Corporation, 2004.

\bibitem[Tut47]{Tut47}
W.~T. Tutte.
\newblock The factorization of linear graphs.
\newblock {\em Journal of the London Mathematical Society}, s1-22(2):107--111,
  1947.
\newblock \href {https://doi.org/10.1112/jlms/s1-22.2.107}
  {\path{doi:10.1112/jlms/s1-22.2.107}}.

\bibitem[Val79]{Val79}
Leslie~G. Valiant.
\newblock Completeness classes in algebra.
\newblock In Michael~J. Fischer, Richard~A. DeMillo, Nancy~A. Lynch, Walter~A.
  Burkhard, and Alfred~V. Aho, editors, {\em Proceedings of the 11h Annual
  {ACM} Symposium on Theory of Computing, April 30 - May 2, 1979, Atlanta,
  Georgia, {USA}}, pages 249--261. {ACM}, 1979.
\newblock \href {https://doi.org/10.1145/800135.804419}
  {\path{doi:10.1145/800135.804419}}.

\bibitem[Wey97]{Wey97}
H.~Weyl.
\newblock {\em The classical groups: their invariants and representations},
  volume~1.
\newblock Princeton University Press, 1946 (1997).
\newblock \href {https://doi.org/10.2307/j.ctv3hh48t}
  {\path{doi:10.2307/j.ctv3hh48t}}.

\bibitem[ZKT85]{ZKT}
V.~N. Zemlyachenko, N.~M. Korneenko, and R.~I. Tyshkevich.
\newblock Graph isomorphism problem.
\newblock {\em J. Soviet Math.}, 29(4):1426--1481, May 1985.
\newblock \href {https://doi.org/10.1007/BF02104746}
  {\path{doi:10.1007/BF02104746}}.

\bibitem[ZLQ18]{ZLQ18}
SM~Zangi, Jun-Li Li, and Cong-Feng Qiao.
\newblock Quantum state concentration and classification of multipartite
  entanglement.
\newblock {\em Physical Review A}, 97(1):012301, 2018.
\newblock \href {https://doi.org/10.1103/PhysRevA.97.012301}
  {\path{doi:10.1103/PhysRevA.97.012301}}.

\end{thebibliography}

\appendix

\end{document}